%% file: paper.tex
\documentclass[letterpaper]{article} %
\usepackage{aaai2026}  %
\usepackage{times}  %
\usepackage{helvet}  %
\usepackage{courier}  %
\usepackage[hyphens]{url}  %
\usepackage{graphicx} %
\usepackage{natbib}  %
\usepackage{caption} %
\DeclareCaptionStyle{ruled}{labelfont=normalfont,labelsep=colon,strut=off} %
\title{Constrained and Robust Policy Synthesis with Satisfiability-Modulo-Probabilistic-Model-Checking}
\author{
    Linus Heck\textsuperscript{\rm 1},
    Filip Macák\textsuperscript{\rm 2},
    Milan Češka\textsuperscript{\rm 2},
    Sebastian Junges\textsuperscript{\rm 1}
}
\affiliations{
    \textsuperscript{\rm 1}Radboud University, Nijmegen, The Netherlands
    \textsuperscript{\rm 2}Brno University of Technology, Brno, Czech Republic
}

\usepackage{csquotes}
\usepackage{standalone}
\usepackage{adjustbox}
\usepackage{subcaption}

\usepackage{amsmath, amssymb, amsthm}

\usepackage{nicefrac}

\usepackage{calc}

\usepackage{paralist}

\usepackage{booktabs}
\usepackage{makecell}

\usepackage{algorithm}
\usepackage{algpseudocode}

\usepackage{marvosym}

\usepackage{paralist}

\usepackage{bm}

\usepackage{forest}
\usepackage{mdframed}

\usepackage{marginnote}
\usepackage[dvipsnames]{xcolor}
\definecolor{bluegray}{rgb}{0.5, 0.65, 0.85}
\definecolor{pastelgreen}{HTML}{58bc82}
\definecolor{greengray}{rgb}{0.5, 0.85, 0.65}
\definecolor{sandstone}{rgb}{0.8, 0.6, 0.2}
\definecolor{ruby}{HTML}{D81E5B}
\definecolor{teal}{HTML}{0E6C6B}
\definecolor{cyan}{HTML}{08B2E3}
\definecolor{purple}{HTML}{52489C}
\definecolor{scblue}{HTML}{1F73A9}
\definecolor{scgray}{HTML}{656A6D}
\definecolor{scred}{HTML}{DA2C38} %

\usepackage{thmtools}
\usepackage{thm-restate}

\definecolor{color1}{RGB}{0, 100, 200}
\definecolor{color2}{RGB}{0, 150, 0}
\definecolor{color3}{RGB}{230, 100, 0}

\usepackage{adjustbox}
\usepackage{tikz}
\usepackage[capitalise]{cleveref}
\usetikzlibrary{arrows.meta, positioning, fit}
\usetikzlibrary{shapes, arrows, calc, patterns}
\usetikzlibrary{automata, shapes.geometric, shapes.symbols}
\tikzset{every state/.style={inner sep=0pt, minimum size=12pt}}
\tikzset{pMC/.style={node distance=0.8cm, on grid, auto, initial text=,every label/.style={}, font=\scriptsize}}

\usepackage{xspace}
\newcounter{goodsmileys}
\newcounter{badsmileys}
\renewcommand{\Smiley}[2]{%
	\begin{tikzpicture}[-,scale=#2]
	\newcommand*{\SmileyRadius}{1.0}%
	\pgfmathsetmacro{\eyeX}{0.5*\SmileyRadius*cos(30)}
	\pgfmathsetmacro{\eyeY}{0.3*\SmileyRadius*sin(30)}
	\draw [line width=0.28mm] (\eyeX-0.25,\eyeY) -- (\eyeX-0.25,\eyeY+0.4);
	\draw [line width=0.28mm] (-\eyeX+0.25,\eyeY) -- (-\eyeX+0.25,\eyeY+0.4);
	\pgfmathsetmacro{\xScale}{2*\eyeX/180}
	\pgfmathsetmacro{\yScale}{1.0*\eyeY}
	\draw[line width=0.28mm, domain=-\eyeX:\eyeX]
	plot ({\x},{
		-0.1+#1*0.15 %
		-#1*1.75*\yScale*(sin((\x+\eyeX)/\xScale))-\eyeY});
	\end{tikzpicture}%
}%
\newcommand{\good}{%
  \stepcounter{goodsmileys}%
  \Smiley{0.6}{0.22}\xspace%
}

\usepackage{pifont}
\newcommand{\cmark}{\ding{51}}%
\newcommand{\xmark}{\ding{55}}%
\newcommand{\yes}{\textcolor{pastelgreen}{\cmark}}
\newcommand{\no}{\textcolor{ruby}{\xmark}}

\newcommand{\dontmatter}{\textcolor{gray!95}{\boldmath$\boldsymbol{\bot}$}}
\newcommand{\viable}{\textsf{viable}}
\newcommand{\blackdontmatter}{{\boldmath$\boldsymbol{\bot}$}}

\mdfdefinestyle{MyFrame}{%
    linecolor=black,
    outerlinewidth=2pt,
    innertopmargin=4pt,
    innerbottommargin=4pt,
    innerrightmargin=6pt,
    innerleftmargin=6pt,
        leftmargin = 2pt,
        rightmargin = 2pt
        }

\newcommand{\pathtostuff}{}

\newcommand{\eq}{\mathbin{=}}
\newcommand{\nneq}{\mathbin{\neq}}

\newcommand{\C}{\mathcal{C}}
\newcommand{\M}{\mathcal{M}}
\newcommand{\D}{\mathcal{D}}

\newcommand{\Ca}{\textbf{C1}}
\newcommand{\Cb}{\textbf{C2}}
\newcommand{\Cc}{\textbf{C3}}
\newcommand{\Cd}{\textbf{C4}}

\newcommand{\conflict}{\(\mathbb{T}_\C\)-conflict}

\newlength\myheight
\newlength\mydepth
\settototalheight\myheight{Xygp}
\newcommand*\inlinegraphics[1]{%
  \settototalheight\myheight{Xygp}%
  \settodepth\mydepth{Xygp}%
  \raisebox{-\mydepth}{\includegraphics[height=\myheight]{#1}}%
}

\usepackage{marginnote}

\usepackage[
    group-digits=integer, group-minimum-digits=4, %
    free-standing-units, unit-optional-argument, %
    ]{siunitx}[=v2]
\usepackage{pgfplots}
\usepackage{longtable}
\usepgfplotslibrary{colorbrewer}
\pgfplotsset{cycle list/Dark2}
\pgfplotsset{
    compat=1.18,
    log ticks with fixed point, %
    table/col sep=tab, %
    unbounded coords=jump, %
    filter discard warning=false, %
    }

\newtheorem{theorem}{Theorem}
\newtheorem{definition}[theorem]{Definition}
\newtheorem{proposition}[theorem]{Proposition}
\newtheorem*{problem}{Problem Statement}

\newtheorem{rnexample}{Running Example}

\begin{document}

\maketitle

\begin{abstract}

The ability to compute reward-optimal policies for given and known finite Markov decision processes (MDPs) underpins a variety of applications across planning, controller synthesis, and verification. However, we often want  policies  (1)~to be robust, i.e., they perform well on perturbations of the MDP and (2)~to satisfy additional structural constraints regarding, e.g., their representation or implementation cost. 
Computing such robust and constrained policies is indeed computationally more challenging.
This paper contributes the first approach to effectively compute robust policies subject to arbitrary structural constraints using a flexible and efficient framework. We achieve flexibility by allowing to express our constraints in a first-order theory over a set of MDPs, while the root for our efficiency lies in the tight integration of satisfiability solvers to handle  the combinatorial nature of the problem and probabilistic model checking algorithms to handle the analysis of MDPs. 
 Experiments on a few hundred benchmarks demonstrate the feasibility for constrained and robust policy synthesis and the competitiveness with state-of-the-art methods for various fragments of the problem.

\end{abstract}

\section{Introduction}
Markov decision processes (MDPs) are a prominent model for sequential decision making under %
uncertainty. For a given MDP, a standard planning task is to compute a policy that optimizes some value, e.g., a discounted reward over an infinite horizon~\cite{Put94}. Optimal policies for such tasks can be efficiently computed both in theory and in practice. 
However, not every policy that can be computed can be deployed: Policies are often subject to additional constraints, in particular (1)~\emph{structural constraints} that require
a certain shape or structure of the \emph{controller} representing the policy and policies should be (2)~\emph{robust} against perturbations of the environments.
Structural constraints may include constraints that require that a policy can be encoded as a small decision tree, see e.g.,~\cite{DBLP:conf/tacas/AshokJKWWY21,DBLP:conf/ijcai/VosV23} and constraints that enforce that the policy is observation-based, as in partially observable MDPs (POMDPs)~\cite{smallwood1973optimal}. Perturbations can be used to express that a policy achieves sufficient value in a set of MDPs, e.g,~\cite{DBLP:conf/aaai/ChadesCMNSB12,DBLP:journals/infsof/GhezziS13}. The key problem statement we work towards is: Is there a policy, subject to structural constraints, such that on every (PO)MDP (environment) in a set of (PO)MDPs, its value is sufficient? Fig.~\ref{fig:robust} illustrates this problem statement. 

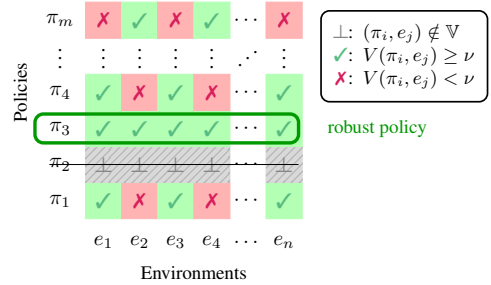
\begin{figure}[t]
    \centering
    \scalebox{0.8}{\input{resources/figures/grid}}
    \caption{
    One axis represents the set of policies $\pi_i$, the other the set of possible environments $e_j$. Some policies (or environments) may not be relevant (\blackdontmatter), e.g. when the policy does not satisfy structural constraints. For all tuples of relevant policy $\pi_i$ and environment $e_j$, the combination performs satisfactory (\cmark) or unsatisfactory (\xmark).
Satisfactory performance is captured by meeting a threshold $\nu$ on the value of the policy on the environment.
The general task is to find a policy $\pi_i$ such that for every environment $e_j$, cell $(i,j)$ is \cmark.
While this figure uses a tabular representation, we use concise symbolic descriptions.}
    \label{fig:robust}
\end{figure}

Accommodating structural constraints and robustness brings two challenges:
(1)~Already when considering \emph{either} structural constraints \emph{or} robustness, we must deal with a \emph{combinatorial} problem.  For various decision problems that underpin planning with additional constraints on the structure of the policy or on the robustness against perturbation, (co)NP-hardness results exist~\cite{andriushchenko2025small, DBLP:conf/tacas/VegtJJ23, DBLP:conf/mfcs/ChatterjeeDH10}. Typically, reasoning about every policy satisfying the constraints or every possible perturbation is necessary \emph{in the worst case}. 
(2)~To support robustness, we must implicitly or explicitly deal with a \emph{quantifier alternation}. The key problem statement existentially quantifies the policy and universally quantifies over the environments.

A classical approach to deal with the combinatorial nature of the problems is to create a mixed integer linear program (MILP), where the integer variables encode the choices for the policies and the perturbations, while an LP encodes the value of the resulting MDP~\cite{DBLP:conf/aips/KumarZ15, DBLP:conf/ijcai/VosV23}. Such monolithic encodings have various drawbacks, in particular, using LP solvers is typically much slower than value- or policy-iteration based approaches to solving MDPs, see e.g.~\cite{DBLP:conf/tacas/HartmannsJQW23}. 

The state-of-the-art, therefore, avoids monolithic encodings 
by applying a separation of concerns~\cite{DBLP:journals/jair/AndriushchenkoCMJK25}.
It instantiates a typical guess-and-check (or inductive synthesis) approach, which we briefly summarize here. 
Consider the black loop in \cref{fig:cegis}: A satisfiability (SAT) solver picks a policy satisfying the structural constraints (or analogously, the environment) and then the policy (and environment) is evaluated by a dedicated numerical routine, in this case, the probabilistic model checking (PMC) framework Storm~\cite{DBLP:journals/sttt/HenselJKQV22}. Without further adaptation, such an approach would need to enumerate all policies. However, Storm can generalize the results from an individual policy to a set of policies~\cite{vcevska2021counterexample}. 
The separation allows using modern SAT solvers for the combinatorial problem and highly efficient routines for policy evaluation. Still, this existing approach has at least two shortcomings. First, generalizing from one policy to a set of policies is often expensive, yet necessary to avoid enumerating the policy space. Second, support for quantifier alternations is not straightforward as the PMC framework would need to evaluate any policy on \emph{every} environment explicitly. For some classes of simple constraints, dedicated branch-and-bound style approaches have been developed with superior performance, e.g., to compute fixed-memory policies for POMDPs~\cite{andriushchenko23search}. 

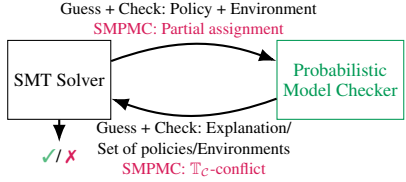
\begin{figure}
    \centering
    \begin{tikzpicture}[
        node distance=2cm and 2cm,
        box/.style={draw, rectangle, minimum width=1.8cm, minimum height=0.5cm, align=center},
        arrow/.style={-{Latex}, thick},
        fitbox/.style={draw, inner sep=0.7cm},
        gplbox/.style={draw=gray, inner sep=1.4cm, dashed},
        every node/.style={scale=0.7}
      ]
        \node[box, style={minimum height=1.5cm}] (smt) {SMT Solver};
        \node[box, style={minimum height=1.5cm}, right=2.2cm of smt, color=ForestGreen] (abstraction) {Probabilistic\\Model Checker};
        \node[above=0.3cm of smt, align=center] (region_r) {};
        \node[above=0.3cm of abstraction, align=center] (big_step) {};
        \node[below=0.3cm of smt] (spec_holds) {\yes / \no};
        \draw[arrow] (smt) -- (spec_holds);
        \draw[arrow, bend left=20] ([yshift=0.25cm] smt.east) to node[above, xshift=-0.1cm, align=center] (arrow1) { \small Guess + Check: Policy + Environment \\ \textcolor{ruby}{\small SMPMC: Partial assignment}} ([yshift=0.25cm] abstraction.west);
        \draw[arrow, bend left=20] ([yshift=-0.25cm] abstraction.west) to node[below, align=center] (arrow2) {\small Guess + Check: Explanation/ \\
         \small Set of policies/Environments\\ \textcolor{ruby}{\small SMPMC: \conflict{}}} ([yshift=-0.25cm] smt.east);
    \end{tikzpicture}
    \caption{A separation of concerns using guess-and-check or counterexample-guided-inductive-synthesis.}
    \label{fig:cegis}
\end{figure}

This paper presents SAT-Modulo-PMC (SMPMC, say \emph{sampick}), a novel synthesis approach that (1) avoids the expensive generalization step crucial for the CEGIS approach, while (2) providing principled support for simultaneous structural constraints and robustness. The approach is (3) sufficiently flexible to support, for the first time, a complete approach that can find provably smallest decision tree encodings of robust policies. 
Technically, the key innovation in SMPMC is a tight integration of modern satisfiability solvers and PMC. More precisely, we consider Z3~\cite{DBLP:conf/tacas/MouraB08}, a satisfiability modulo theories solver with support for quantified logics and the recently developed option to add custom theories~\cite{DBLP:conf/vmcai/BjornerEK23}. We develop a custom theory based on probabilistic model checking. Algorithmically, SMPMC can be seen as a powerful branch-and-bound variation where the SAT solver realizes intelligent branching and the probabilistic model checker ensures effective bounding.  Despite its flexibility, the novel technique is highly competitive and it outperforms a monolithic approach to robust and constrained synthesis by orders of magnitude.

\paragraph{Contributions.}
In summary, we present SMPMC, a novel, versatile, and efficient approach that tightly integrates satisfiability solvers with probabilistic model checking. SMPMC is the first effective approach that constructs decision tree representations of robust policies and is even able to prove the absence of a fixed-size robust decision tree. For various fragments of the robust constrained synthesis problem,
SMPMC is generally at least competitive with ad-hoc state-of-the-art heuristics and solves many benchmarks that the state-of-the-art cannot solve.

\section{Setting and Motivating Example}
\begin{figure}
    \centering
    \begin{subfigure}{0.45\linewidth}
        \centering
        \begin{adjustbox}{height=2cm}
        \input{resources/figures/bug}
        \end{adjustbox}
        \caption{Single initial location}
        \label{fig:bug1}
    \end{subfigure}
    \begin{subfigure}{0.45\linewidth}
        \centering
        \begin{adjustbox}{height=2cm}
        \input{resources/figures/robustbug}
        \end{adjustbox}
        \caption{Multiple initial locations}
        \label{fig:bug2}
    \end{subfigure}
    \caption{Running example: Beetle wants to reach the target.}
    \label{fig:bug}
\end{figure}
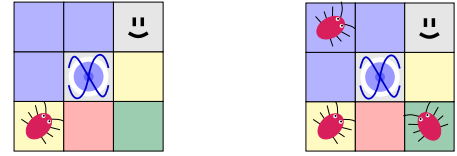

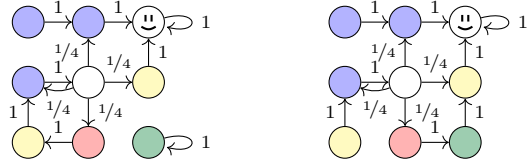
\begin{figure}
    \centering%
    \begin{subfigure}{0.49\linewidth}
        \centering%
            \begin{tikzpicture}[pMC]
                \node[state, fill=yellow!30] (s00) {};
                \node[state, fill=red!30] (s01) [right of=s00] {};
                \node[state, fill=ForestGreen!40] (s02) [right of=s01] {};
                
                \node[state, fill=blue!30] (s10) [above of=s00] {};
                \node[state] (s11) [right of=s10] {};
                \node[state, fill=yellow!30] (s12) [right of=s11] {};
                
                \node[state, fill=blue!30] (s20) [above of=s10] {};
                \node[state, fill=blue!30] (s21) [right of=s20] {};
                \node[state] (s22) [right of=s21] {$\good$};
                
                \path[->] (s00) edge node[left] {$1$} (s10);
                
                \path[->] (s20) edge node[above] {$1$} (s21);
                \path[->] (s10) edge node[above] {$1$} (s11);
                \path[->] (s21) edge node[above] {$1$} (s22);
                \path[->] (s01) edge node[above] {$1$} (s00);
                
                \path[->] (s02) edge[loop right] node[right] {$1$} (s02);
                \path[->] (s12) edge node[right] {$1$} (s22);
                
                \path[->] (s11) edge node[left] {$\nicefrac{1}{4}$} (s21);
                \path[->] (s11) edge node[right] {$\nicefrac{1}{4}$} (s01);
                \path[->] (s11) edge node[above] {$\nicefrac{1}{4}$} (s12);
                \path[->, bend left=20] (s11) edge node[below] {$\nicefrac{1}{4}$} (s10);
                
                \path[->] (s22) edge[loop right] node[right] {$1$} (s22);
            \end{tikzpicture}
            \caption{\(\{d_r\eq 1, d_g\eq 3, d_b \eq 3, d_y \eq 0\}\) solves single initial location}
            \label{fig:bugmc1}
        \end{subfigure}
    \begin{subfigure}{0.49\linewidth}
        \centering%
            \begin{tikzpicture}[pMC]
                \node[state, fill=yellow!30] (s00) {};
                \node[state, fill=red!30] (s01) [right of=s00] {};
                \node[state, fill=ForestGreen!40] (s02) [right of=s01] {};
                
                \node[state, fill=blue!30] (s10) [above of=s00] {};
                \node[state] (s11) [right of=s10] {};
                \node[state, fill=yellow!30] (s12) [right of=s11] {};
                
                \node[state, fill=blue!30] (s20) [above of=s10] {};
                \node[state, fill=blue!30] (s21) [right of=s20] {};
                \node[state] (s22) [right of=s21] {$\good$};
                
                \path[->] (s00) edge node[left] {$1$} (s10);
                
                \path[->] (s20) edge node[above] {$1$} (s21);
                \path[->] (s10) edge node[above] {$1$} (s11);
                \path[->] (s21) edge node[above] {$1$} (s22);
                \path[->] (s01) edge node[above] {$1$} (s02);
                
                \path[->] (s02) edge node[right] {$1$} (s12);
                \path[->] (s12) edge node[right] {$1$} (s22);
                
                \path[->] (s11) edge node[left] {$\nicefrac{1}{4}$} (s21);
                \path[->] (s11) edge node[right] {$\nicefrac{1}{4}$} (s01);
                \path[->] (s11) edge node[above] {$\nicefrac{1}{4}$} (s12);
                \path[->, bend left=20] (s11) edge node[below] {$\nicefrac{1}{4}$} (s10);
                \path[->] (s22) edge[loop right] node[right] {$1$} (s22);

            \end{tikzpicture}
            \caption{\(\{d_r\eq 3, d_g\eq 0, d_b \eq 3, d_y \eq 0\}\) solves multiple initial locations}
            \label{fig:bugmc2}
        \end{subfigure}
        \caption{Markov chains induced by different beetle policies.}
    \end{figure}

In this section, we introduce a motivating example that we refer to throughout this paper. We first introduce the problem and then work towards a formalization.

\paragraph{Problem description.}
Consider Fig.~\ref{fig:bug}: The agent (depicted as a beetle) moves in a tiny grid-world environment including the target in the north east corner, the only cell equipped with a reward,
and a fountain in the middle, which will spray the beetle in a random direction.
We impose the following simplified constraints: (1) the beetle must move in the same direction in cells with the same floor color, and (2) must go in a different direction on blue tiles than on yellow tiles. Our goal is to find a memoryless policy that eventually 
collects the reward. 
\cref{fig:bug1} illustrates this problem for a known initial location of the agent. This problem corresponds to \emph{structurally constrained} synthesis of (memoryless) policies in (PO)MDPs. 
If the agent must use a single policy that is robust against perturbations of the initial state (see \cref{fig:bug2}), the problem becomes harder
and corresponds to a constrained \emph{and robust} policy synthesis in (PO)MDPs. 

\paragraph{Challenges.}
Two aspects make the example challenging: 
\begin{compactenum}
    \item Some policy choices (e.g., the beetle goes right on blue tiles) can affect multiple state-action pairs in the model (e.g., the direction for all blue tiles).
    \item Some constraints on policy or environment choices (e.g., the beetle goes in different directions on blue and yellow tiles) involve multiple choices at once.
\end{compactenum}
More generally, the first aspect is crucial when reasoning about robust policies for multiple environments, where the agent cannot observe which environment is describing the world. In those cases, a robust policy must take the same actions in all environments. This adds structurally simple identities between action choices in different states, similar to observation-based memoryless policies in POMDPs. 
Beyond these structurally simple identities, the second property is a simple type of additional constraint. 

\paragraph{Problem modeling.}
We reason about the policy and environment in a unified parameter space, where policies are given using \emph{controllable parameters} \(X = \{d_r,d_g,d_b,d_y\}\) (i.e., the direction for each color). We can impose constraints such as \(d_b \nneq d_y\). To represent the environments, we use \emph{uncontrollable parameters} \(Y \eq \{s_x, s_y\}\) (i.e., representing the starting coordinates).
The beetle's policy must be robust against the choice of the parameters in \(Y\). The possible starting locations are restricted by the constraint $(s_x\eq0 \land s_y \eq 0) \lor (s_x \eq 2 \land s_y \eq 0) \lor (s_x \eq 0 \land s_y \eq 2)$. In general, constraints can be represented by arbitrary logical formulas over $X \cup Y$. A satisfying MC (corresponding to a satisfying assignment) for the single initial location can be seen in \cref{fig:bugmc1}. In the case of multiple initial locations, we have a quantifier alternation: we ask whether there exists an assignment of the controllable parameters that is satisfying for all assignments of the uncontrollable parameters. A satisfying MC (where a reward $1$ is collected from each initial state) is shown in \cref{fig:bugmc2}.

\paragraph{Notation.} We present our approach using generic (un)controllable parameters $X, Y$ and assignments rather than using policies and environments.

\section{Problem Statement}
\label{sec:statement}

We formalize the main problem statement. We first recap Markov decision processes (MDPs) and colored MDPs.

\begin{definition}[Markov Decision Process]
    A \emph{Markov decision process (MDP)} is a tuple \( \M = (S, s_I, Act, \mathcal{P}, R)\) with a finite set \(S\) of states, an initial state \( s_I \in S \), a finite set  \(Act\) of actions, a partial transition function \(\mathcal{P}\colon S \times Act \rightharpoonup Distr(S)\) and a state reward function \(R : S \rightarrow \mathbb{R}_{\geq 0}\).
\end{definition}
We write \(Act(s) := \{\alpha \in Act \mid \mathcal{P}(s, \alpha) \neq \bot\}\) for \emph{ available actions} at \(s\).
An MDP policy is a function \(\pi\colon S \rightarrow Act\) s.t. \(\pi(s) \in Act(s)\) for all \(s \in S\).
A Markov chain (MC) is an MDP where \(|Act(s)|=1\) for all \(s \in S\). We denote their transition functions as \(\mathcal{P} \colon S \rightarrow Distr(S)\). 
For MDP \(\M\) and policy \(\pi\), we define the induced MC \(\M[\pi] = (S, s_I, \mathcal{P}', R)\) with \(\mathcal{P}'(s) = \mathcal{P}(s, \pi(s))\) for any \(s \in S\).

\paragraph{Specification.} We consider undiscounted indefinite-horizon expected rewards~\cite{Put94}. These generalize both finite-horizon rewards and infinite-horizon discounted rewards.
Given an MC \(\D\), a \emph{path} \(\xi\) in \(\D\) is an infinite sequence of states \(s_0 s_1 s_2 \ldots \in S\) where \(s_0 = s_I\). We define \(R(\xi) = \sum_{i=0}^{\infty} R(s_i)\) as the (possibly infinite) reward of a path. We define the \emph{value} of $D$ as \(V(\D) := \mathbb{E} [R(\xi)]\) if the expected value exists and $\infty$ otherwise.
For an MDP \(\M\), we define the \emph{maximal value} \(V^{\max}(\M) := \max_\pi V(\M[\pi])\) and the \emph{minimal value} \(V^{\min}\) analogously.

\paragraph{Parameter space.}
We introduce a \emph{parameter space} to symbolically reason about policies and environments. Let \(X\) be a set of \emph{controllable parameters} and \(Y\) be a set of \emph{uncontrollable parameters}, representing, e.g., policies and environments, respectively. The sets $X$ and $Y$ are disjoint. 
We denote assignments to variables \(X \cup Y\) by values from \(\mathbb{Z}\) with \(\mathbb{Z}^{X \cup Y}\). We consider a finite \emph{constrained parameter space} \(\mathbb{V} \subseteq \mathbb{Z}^{X \cup Y}\).  In the remainder, $\theta$ denotes assignments in $\mathbb{V}$, while $\eta \subseteq \mathbb{V}$ denotes a nonempty subset of~\(\mathbb{V}\) and is referred to as a \emph{set of assignments}. We write \(\eta(p) = \{\theta(p) \mid \theta\in \eta\}\) for parameter \(p \in X \cup Y\). In this paper, we represent \(\mathbb{V}\) as the satisfying assignments of a logical formula.

\paragraph{Projected assignments.}
For $Z \in \{X, Y\}$ and $\theta \in \mathbb{V}$, the projected assignment $\theta_Z \in \mathbb{Z}^Z$ coincides with $\theta$ on the variables $Z$. We use such projections to write assignment $\theta$ as $\theta_X\theta_Y$ and lift them to the parameter space: \(\mathbb{V}_X := \{\theta_X \in \mathbb{Z}^X \mid \exists \theta_Y:  \theta_X\theta_Y \in \mathbb{V}\}\). \(\mathbb{V}_Y\) is defined analogously. 

\paragraph{Colored MDPs.} We represent finite sets of MCs by a \emph{colored MDP} \cite{DBLP:journals/jair/AndriushchenkoCMJK25}. A colored MDP~$\C$ relates to the constrained parameter space \(\mathbb{V}\) by mapping state-action pairs $(s,a)$ in $\C$ to a set of assignments $\eta \subseteq \mathbb{V}$. This represents for which $\theta \in \eta$ distribution $(s,a)$ exists:

\begin{definition}[Colored MDP]
    A \emph{colored MDP} \(\C = (\M, \mathbb{V}, \kappa)\) consists of an MDP \(\M\) with states $S$ and actions $Act$, a constrained parameter space \(\mathbb{V}\), and a \emph{coloring} \(\kappa(s,\alpha) \subseteq \mathbb{V}\), such that for all \(\theta \in \mathbb{V}, s \in S\), there is exactly one action \(\alpha \in Act(s)\) with \(\theta \in \kappa(s, \alpha)\).
\end{definition}

\begin{rnexample}
    For \cref{fig:bug1}, suppose \(s \in S\) is a blue state with \(Act(s) = \{0, 1, 2, 3\}\) representing the directions. Then \(\kappa(s, \alpha) = \{\theta \mid \theta(d_b) = \alpha\}\) for \(\alpha \in Act(s)\). In \cref{fig:bugmc1}, we have \(\theta(d_b)=3\), which selects action \(3\) in all blue states.
\end{rnexample}

Given a colored MDP \(\C\) on \(\mathbb{V}\) and a set of assignments \(\eta\), the \emph{induced MDP} \(\C[\eta]\) restricts \(\C\) to actions admissible by \(\eta\):

\begin{definition}[Induced MDP]
    \label{def:inducedmdp}
    For a set of assignments \(\eta \subseteq \mathbb{V}\), the \emph{induced MDP} \(\C[\eta]\) is the largest sub-MDP \(\M'\) of \(\M\) with actions \(Act'\) s.t. \(\eta \cap \kappa(s, \alpha) \neq \emptyset\) for all \(s \in S, \alpha \in Act'(s)\).
    For \(\theta \in \mathbb{V}\), the MDP \(\C[\theta] := \C[\{\theta\}]\) is an MC.
\end{definition}
We can now formally state the problem this paper tackles:

\begin{mdframed}[style=MyFrame,nobreak=true]
\textbf{Problem Statement} (Robust Constrained Feasibility).
Given a colored MDP \(\C\) with constrained parameter space \(\mathbb{V}\) on controllable parameters \(X\) and uncontrollable parameters \(Y\), and a threshold \(\nu \in \mathbb{R}\), does there exist an assignment \(\theta_X \in \mathbb{V}_X\) s.t for all assignments \(\theta_Y \in \mathbb{V}_Y\) with \(\theta_X  \theta_Y \in \mathbb{V}\): \(V(\C[\theta_X \theta_Y]) \geq \nu\)?
\end{mdframed}

We thus ask for a policy, encoded as \(\theta_X\), that satisfies the value threshold \(\nu\) in all environments, encoded by \(\theta_Y\).
In the remainder, we assume a fixed \(\nu\) and omit it for conciseness.

\section{The SMPMC Approach}
\label{section:intosmt}

We introduce \emph{Satisfiability Modulo Probabilistic Model Checking} (SMPMC, pronounced \emph{sampick}). We first define a first-order theory that encodes that the value of a colored MDP at an assignment is satisfactory and show that the validity of specific sentences in this theory corresponds to solving the problem statement. We then discuss how our novel SMPMC approach integrates this theory into the lazy constraint-driven conflict learning (CDCL) method. Finally, we describe how to extend this method to solve robust constrained synthesis.\footnote{
We provide an extended example in Appendix B.1 and prove all theorems in Appendix B.2 (\url{https://arxiv.org/abs/2511.08078}).
}

\subsection{First-order formulation}
Given an instance of the problem statement on a colored MDP \(\C\) and let \(X \cup Y = \{p_1, p_2, \ldots, p_n\}\). We formulate the problem statement in a dedicated first-order theory \(\mathbb{T}_\C\) over bounded integer variables $X \cup Y$, which introduces a single $n$-ary predicate \(\text{\viable}\) over \(X \cup Y\).
Intuitively, the atom $\viable(\theta(p_1),\dots, \theta(p_n))$ evaluates to true iff $V(\C[\theta]) \geq \nu$.  Formally, we uniquely characterize the theory \(\mathbb{T}_\C\) by the set of true sentences. The true sentences are the deductive closure of the following sets of sentences:
\begin{align*}
    &\{\viable(\theta(p_1), \ldots, \theta(p_n)) &\mid \theta \in \mathbb{V} \text{ with } V(\C[\theta]) \geq \nu\},\\
    & \{\lnot\viable(\theta(p_1), \ldots, \theta(p_n)) &\mid \theta \in \mathbb{V} \text{ with } V(\C[\theta]) < \nu\}.
\end{align*}
In particular, this ensures that for an assignment \(\theta \in \mathbb{V}\), the first-order sentence \(\viable(\theta(p_1), \ldots, \theta(p_n))\) is true in \(\mathbb{T}_\C\) iff the colored MDP \(\C\) has a value of at least \(\nu\) at \(\theta\).

By including the theory of the bounded integers \cite[p831]{DBLP:series/faia/2009-185}, we can use a formula $\tau_\mathbb{V}$ to represent the constraints of the constrained parameter space \(\mathbb{V}\) logically. Given the logical formula \(\tau_\mathbb{V}\), the problem statement can be rephrased as a first-order sentence:
\begin{align*}
    \Phi := \; &\exists x_1 \ldots \exists x_n \; \forall y_1 \ldots  \forall y_m:\tau_\mathbb{V}(x_1, \ldots, x_n, y_1, \ldots, y_m)\\ \rightarrow \; &\viable(x_1, \ldots, x_n, y_1, \ldots, y_m).
    \tag{1} \label{eq:problem}
\end{align*}

\begin{restatable}{theorem}{validtheorem}
    \label{theorem:valid}
    For a colored MDP \(\C\), the assignment \(\theta_X \in \mathbb{V}_X\) over variables \(X = \{x_1, \ldots, x_n\}\) is a satisfying assignment of the existential quantifier in the sentence \(\Phi\) if and only if \(\theta_X\) is a solution to the problem statement.
\end{restatable}

\subsection{CDCL and theory solvers}

We first describe SMPMC assuming that \(Y = \emptyset\) and leave the evaluation of the universal quantifier for \cref{subseq:universal}.
The core algorithm in most SMT solvers is \emph{Conflict-Driven Clause Learning (CDCL)} \cite{DBLP:conf/aaai/BayardoS97,DBLP:journals/tc/Marques-SilvaS99}.
The goal of CDCL is to find a satisfying assignment of the existentially quantified variables in our input formula \(\Phi\). To do this, CDCL iteratively constructs a partial assignment \(K\) of these variables\footnote{Strictly speaking, our variables are bit vectors. We then say a variable is assigned when all bits are assigned.}.
At each step of this iteration, CDCL invokes a subroutine to check whether any full extension of \(K\) could satisfy \(\Phi\). 
If this is not the case, we call \(K\) a \emph{conflict}. In CDCL with theories \cite{DBLP:journals/jacm/NieuwenhuisOT06}, all theories involved in the assignment \(K\) provide a \emph{theory solver} that is queried and asked whether \(K\) is a conflict within that theory.  In SMPMC, we provide a theory solver for the theory \(\mathbb{T}_\C\) introduced above. \cref{fig:cegis} describes the loop in red: (1)~The SMT solver provides the theory solver with \(K\). (2)~The theory solver reports whether \(K\) is a conflict.

We match a partial assignment \(K\) to a set of first-order \emph{literals} that assign variables to values. In the following, we use both interchangeably. A partial assignment \(K\) is \emph{consistent} if there exists a full variable assignment \(\theta\) that makes all literals in \(K\) true, i.e., \(\theta\) is a full extension of the partial assignment \(K\). 
The set \(K\) is \emph{inconsistent} if there is no such variable assignment.
We call a partial assignment \(K\) a \emph{\(\mathbb{T}_\C\)-conflict} if \(K \cup \mathbb{T}_\C\) is inconsistent. If the theory solver proves \(K\) to be a \conflict{}, it propagates a conflict back to the SMT solver. It repeats this process until it finds a satisfying assignment or proves \(\Phi\) false.
We illustrate the CDCL algorithm combined with our theory solver in the following~example.

\begin{rnexample}
    \cref{fig:cdcl} shows a hypothetical run of CDCL where it has fixed three variables. The partial assignment is: $K =\{\viable(d_r,d_g,d_b,d_y), d_y\eq 0,d_r\eq 2,d_b\eq 1\}$.
    This means that in \cref{fig:bug1}, the beetle must go north on yellow tiles, south on red tiles, and west on blue tiles. The theory solver proves that \(K\) is a \conflict{}.
\end{rnexample}
From now on, \(K\) contains only \((\lnot)\viable(p_1, \ldots, p_n)\) and assignments of the variables \(p_1, \ldots, p_n \in X \cup Y\). The assignment \(\theta \in \mathbb{V}\) satisfies \(K\) if it satisfies its conjunction.

\begin{figure}
    \centering
    \begin{tikzpicture}[level distance=1cm, sibling distance=1cm, edge from parent/.style={draw, ->, thick}]
        \node (mainroot) {\small root}
            child {node (root) [circle, draw, thin] {\footnotesize $d_y$}
            child {node {$\ldots$}}
            child {node {$\ldots$}}
            child {node (x4) [circle, draw, thin] {\footnotesize $d_r$}
                child {node {$\ldots$}}
                child {node (conflict) {\LARGE\Lightning}
                edge from parent node[right, yshift=-0.0cm] {\small 1}
                }
                child {node {$\ldots$}}
            edge from parent node[right, yshift=0.2cm] {\small 2}
            }
            edge from parent node[left, yshift=0.1cm] {\small 0}
        }
        child {node {$\ldots$}};
        \node[anchor=west, fill=yellow!30] (sign1) at ([xshift=2cm]root.east) {\fbox{\small $d_y\eq 0$ (go north)}};
        \node (sign0) at ([yshift=1cm]sign1) {\fbox{\small $\viable(d_r,d_g,d_b,d_y)$}};
        \node[fill=red!30] at ([yshift=-1cm]sign1) {\fbox{\small $d_r\eq 2$ (go south)}};
        \node[fill=blue!30] at ([yshift=-2cm]sign1) {\fbox{\small $d_b\eq 1$ (go west)}};
    \end{tikzpicture}
    \caption{Sketch of CDCL.}
    \label{fig:cdcl}
\end{figure}
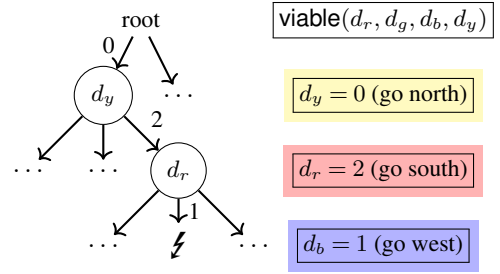

\subsection{Theory solver for \(\mathbb{T}_\C\)}

The theory solver wants to prove that a partial assignment \(K\) is a \conflict{}. As \(\mathbb{V}\) is finite, one correct but intractable approach to decide whether a given \(K \cup \mathbb{T}_\C\) is (in)consistent is to enumerate all possible assignments. Instead, the theory solver outlined in this section checks an induced MDP as an \emph{abstraction} of the colored MDP \(\C\). If the value of the induced MDP is not within the threshold \(\nu\), it derives a \conflict{}:

\begin{restatable}{theorem}{inconsistent}
    \label{prop:inconsistent}
    Let \(\C\) be a colored MDP and \(K\) be a partial assignment to \(\Phi\) s.t. \(v := \viable(p_1, \ldots, p_n) \in K\). Let
    \(\eta = \{\theta \mid \theta \text{ satisfies } K \setminus \{v\}\}\).
    If \(V^{{\max}}(\C[\eta]) < \nu\), then \(K\) is a \conflict.
    The same is true for \(\lnot v \in K\) and \(V^{{\min}}(\C[\eta]) \geq \nu\).
\end{restatable}
\cref{prop:inconsistent} follows from 
\cite{DBLP:journals/jair/AndriushchenkoCMJK25}. The intuition is that if \(V^{{\max}}(\C[\eta]) < \nu\), then all induced MCs of the MDP \(\C[\eta]\) have a value below \(\nu\). Thus, there is no assignment \(\theta \in \eta\) such that \(V(\C[\theta]) \geq \nu\). Note that because deciding whether \(K\) is a \conflict{} is NP-hard, \cref{prop:inconsistent} only gives us a complete method if \(\eta\) contains a single \(\theta\).

SMT solvers profit from theory solvers that find small conflicts \cite[p849]{DBLP:series/faia/2009-185}. In particular, a \conflict{} \(K\) may contain statements that do not contribute to being one. Suppose that we know that \(K\) is a \conflict{}, we now aim to find a smaller \conflict{} \(C(K) \subseteq K\).

We say that a coloring \(\kappa(s, \alpha)\) is \emph{independent} of a parameter \(p_i\) if the value of that parameter is irrelevant for membership in \(\kappa(s, \alpha)\). If \emph{all} colorings in \(\C\) that are reachable in \(\C[\eta]\) are independent of \(p_i\), it cannot contribute to \(K\) being a \conflict{}. We can thus remove all assignments of independent parameters from \(K\), yielding a \emph{restricted partial assignment} \(C(K)\) over the remaining assignments. We define independence and \(C(K)\) precisely in Appendix B.2.
\begin{restatable}{theorem}{smallerconflicts}
    \label{theorem:smallerconflicts}
    Given %
    \(\C\) and %
    \(K\) from \cref{prop:inconsistent}, the restricted partial assignment \(C(K) \subseteq K\) is also a \conflict{}.
\end{restatable}
\cref{theorem:smallerconflicts} describes a computationally efficient way to compute a smaller \conflict{} by computing the independent parameters. In contrast, the \emph{counterexample generalization} on full assignments in \cite{vcevska2021counterexample} is NP-hard.

\begin{rnexample}
    In \cref{fig:cdcl}, the theory solver checks whether for \(\eta = \{\theta \mid \theta(d_y)\eq0, \theta(d_r)\eq 2, \theta(d_b)\eq 1\}\), we have $V^{\max}(\C[\eta]) < 1$. This is true, so \(K\) is a \conflict{} because of \cref{prop:inconsistent}. As all colorings reachable in $\C[\eta]$ are independent of $d_r$, the theory solver can compute a smaller \conflict{} \(C(K) = \{\viable(d_r,d_g,d_b,d_y), d_y\eq 0, d_b\eq 1\}\).
\end{rnexample}

\subsection{Handling quantifier alternation}
\label{subseq:universal}

We now consider that \(Y \neq \emptyset\). Given the quantified statement \(\Phi\) from \cref{eq:problem}, the SMT solver uses \emph{model-based quantifier instantiation (MBQI)} to instantiate the universally quantified variables until it either finds a satisfying assignment for the entire formula or proves it false \cite{DBLP:conf/cav/GeM09}.

Given \(\Phi\) from \cref{eq:problem}, MBQI first finds a candidate variable assignment \(\theta_X\) for the existentially quantified variables \(x_1, \ldots, x_n\). 
Then, it substitutes these values into the quantifier-free formula \(\phi := \tau_{\mathbb{V}}(\ldots) \rightarrow \viable(\ldots)\), yielding the formula \(\phi^{\theta_X}(y_1, \ldots, y_m)\) over the variables in $Y$.
Next, MBQI checks whether \(\lnot \phi^{\theta_X}\) is inconsistent.
If it is inconsistent, then \(\theta_X\) is a satisfying assignment of \(\Phi\)'s existential quantifier, as there is no instantiation of \(y_1, \ldots, y_n\) that makes \(\phi^{\theta_X}\) false. Otherwise, there is a satisfying variable assignment \(\theta_Y\) of \(\lnot\phi^{\theta_X}\). MBQI adds a new conflict $\{\phi(x_1^{\theta_X}, \ldots x_n^{\theta_X}, y_1^{\theta_Y}, \ldots, y_n^{\theta_Y})\}$, which rules out at least $\theta_X$ (and might rule out more). It repeats this process until either a satisfying assignment is found or \(\Phi\) is proven false.

\begin{rnexample}
    Given the environment in \cref{fig:bug2}, suppose CDCL chooses \(\theta_X = \{d_r\eq 1, d_g\eq 3, d_b \eq 3, d_y \eq 0\}\). MBQI checks whether \(\lnot \phi^{\theta_X}(s_x, s_y)\) is inconsistent. 
    As the beetle starting in the bottom right corner will not reach the target in the MC depicted in \cref{fig:bugmc1},
    \(\{s_x\eq 2,s_y\eq 0\}\) satisfies \(\lnot \phi^{\theta_X}(s_x, s_y)\). MBQI pushes the conflict \(\{\phi(1,3,3,0,2,0)\}\), ruling out \(\theta_X\).
    CDCL chooses \(\theta_X' = \{d_r\eq 3, d_g\eq 0, d_b\eq 3, d_y\eq 0\}\) (see \cref{fig:bugmc2}). As \(\lnot \phi^{\theta_X'}(s_x, s_y)\) is inconsistent, the beetle will reach the target no matter where it starts with \(\theta_X'\).
\end{rnexample}

\section{Experiments}

We investigate the performance of the proposed policy synthesis algorithm SMPMC. The implementation and all the considered benchmarks are available\footnote{Code, benchmarks: \url{https://doi.org/10.5281/zenodo.17573007}}.

\begin{table*}
\centering
\begin{tabular}{@{}lcccll@{}}
\toprule
 & S.C.\({}^{*}\) & $\tau_{\mathbb{V}} \neq \top$ & $Y \neq \emptyset$ & \textbf{Problem Statements: Is there a...} & \textbf{Baselines} \\
\midrule
\textbf{C1} & \yes & \yes & \yes &
\makecell[l]{
decision tree policy that works in all environments?
\vspace{0.2em}
}
& \texttt{SMT(LRA)} \\
\textbf{C2} & \yes & \yes & \no &
\makecell[l]{
decision tree policy in a POMDP? \\
satisfying, cost-bounded MC in a set of MCs?
\vspace{0.2em}
} & \texttt{PAYNT-CEGIS}, \texttt{SMT(LRA)} \\
\textbf{C3} & \yes  & \no & \yes &
\makecell[l]{
policy in a POMDP that works in all environments? \\
satisfying MC that is robust against perturbations?
\vspace{0.2em}
} & \texttt{\texttt{PAYNT-AR}} (extended), \texttt{SMT(LRA)} \\
\textbf{C4} & \yes & \no & \no &
\makecell[l]{
policy in a POMDP? \\
satisfying MC in a set of MCs?
} & \texttt{PAYNT-\{AR,CEGIS\}}, \texttt{SMT(LRA)} \\
\bottomrule
\end{tabular}
\caption{Summary of configurations, problem statements, and baselines we compare with. The policies are memoryless or fixed-memory. \({}^{*}\)All problem categories exhibit structural constraints, e.g., in the form of partial observability.}
\label{tab:configurations}
\end{table*}

\begin{table*}
\centering
\begin{tabular}{@{}lrrrrrr@{}}
    \toprule
    & \#Benchmarks & SMPMC (New!) & \texttt{\texttt{PAYNT-AR}} & \texttt{PAYNT-CEGIS} & \texttt{PAYNT-Hybrid} & \texttt{SMT(LRA)} \\
    \midrule
    \textbf{C1} & 112 & \textbf{77 (53)} & N/A & N/A & N/A & 24 (0) \\
    \textbf{C2} & 72 & \textbf{54 ~~(1)} & N/A & \textbf{54 ~(1)} & N/A & 31 (0) \\
    \textbf{C3} & 82 & \textbf{43 (16)} & \textbf{43 (16)} & N/A & N/A & 7 (0) \\
    \textbf{C4} & 106 & \textbf{88} ~~(5) & 86 ~~(\textbf{8}) & 24\({}^{*}\) (0) & 23\({}^{*}\) (0) & 30 (0) \\
    \(\sum\)    & 372  & \textbf{262 (75)} & 129 (24) & 78 ~(1) & 23 ~(0) & 92 (0) \\
    \bottomrule
\end{tabular}
\caption{The total number of solved benchmarks with the 30-minute timeout for each benchmark. The numbers in parentheses show the number of benchmarks that no other tool has solved. N/A indicates that the method does not support the problem class. \({}^{*}\)Due to limitations in counterexample generalization, 57 benchmarks are not supported by \texttt{PAYNT-CEGIS}, \texttt{Paynt-Hybrid}.}
\label{tab:results}
\end{table*}

\paragraph{Research questions.} We focus on the following questions:
\begin{compactenum}
    \item[Q1:] \emph{Can SMPMC solve a class of synthesis problems that goes beyond the capabilities of existing tools?}

    \item[Q2:] \emph{Is SMPMC competitive with 
    state-of-the-art methods  
    for
various fragments of the synthesis problem?}
\end{compactenum}

\paragraph{SMPMC implementation.}
We implemented SMPMC on top of the probabilistic model checker Storm~\cite{DBLP:journals/sttt/HenselJKQV22} and the SMT solver Z3~\cite{DBLP:conf/tacas/MouraB08}, leveraging Z3's user propagator API~\cite{DBLP:conf/vmcai/BjornerEK23}. 
We use PAYNT~\cite{DBLP:journals/jair/AndriushchenkoCMJK25} to parse the input. Appendix A gives details.

\paragraph{Benchmarks.}
To systematically answer the research questions, we define four classes \Ca--\Cd{} of synthesis problems, see \cref{tab:configurations}. They represent different, NP-hard, 
fragments of the problem statement. The classes \Cb--\Cd{} are special cases due to
additional assumptions on the constraint \(\tau_{\mathbb{V}}\) and the uncontrollable parameters \(Y\). 
We consider 372 benchmarks from different domains: 
(1)~sets of MCs, given by parametrized finite-state probabilistic programs, mostly from ~\cite{DBLP:journals/jair/AndriushchenkoCMJK25}, 
(2)~synthesis of fixed-memory finite-state controllers (FSCs) for POMDPs from~\cite{andriushchenko23search},
(3)~multiple-environment (ME) MDPs from \cite{DBLP:conf/atva/AndriushchenkoCJM24,DBLP:conf/tacas/VegtJJ23}
that model variations of a system in different environments~\cite{DBLP:conf/aaai/ChadesCMNSB12},
and
(4)~ME-POMDPs from~\cite{hmpomdps}. Finally, (5)~we extended some standard POMDP from (2) to ME-POMDPs.
These benchmarks were translated to colored MDPs akin to \cite{DBLP:journals/jair/AndriushchenkoCMJK25}.
Our benchmarks include a similar number of satisfiable and unsatisfiable instances.

\paragraph{Setup.} All experiments ran on a single core of an AMD~Ryzen~TRP~5965WX with a 30-minute timeout and 16GB of available memory. The considered baselines are explained in the \emph{baselines} paragraphs in Q1 and Q2 sections. All tools use optimistic value iteration \cite{DBLP:conf/cav/HartmannsK20} with a precision of \(10^{-4}\) for PMC. We ran each experiment three times and confirmed that \cref{tab:results} is the same. All of the considered methods are deterministic and have insignificant variance in timing on fixed hardware.

\begin{mdframed}[style=MyFrame,nobreak=true]
\textbf{Main result:}
\cref{tab:results} reports the total number of solved benchmarks across all considered methods and problem classes.  In \Ca, only SMPMC solves a significant part of the benchmarks, and for all other classes, it is competitive and solves some previously unsolvable instances. 
\end{mdframed}
Below, we detail the results in the context of Q1 and Q2.

\subsection*{Q1: Going beyond capabilities of existing tools}

\textbf{Setting.} Inspired by the recent efforts to find robust policies~\cite{hmpomdps} and interpretable policies represented by small decision trees~\cite{DBLP:conf/ijcai/VosV23}, we combine both these requirements and seek for robust decision trees (DTs) in various ME-(PO)MDP benchmarks. We formulate an SMT constraint \(\tau_{\mathbb{V}}\) that specifies that the policy can be represented by a DT with \(n\) nodes (the encoding is described in Appendix D.1).
We consider  \(n \in \{1,3, \ldots, 15\}\) for all benchmarks.

\smallskip
\noindent
\textbf{Baseline.} To our best knowledge, there is no tool that is able to solve \Ca~problems. Therefore, we implemented a baseline, further denoted as \texttt{SMT(LRA)}, that supports these problems by constructing a monolithic SMT query with \emph{linear real arithmetic} and running Z3.%

\smallskip
\noindent \textbf{Results.} 
\cref{tab:results} shows that SMPMC clearly outperforms the \texttt{SMT(LRA)} baseline: It solves all instances that \texttt{SMT(LRA)} does and it additionally solves 53 benchmarks that are unfeasible for \texttt{SMT(LRA)}. More detailed results presented in~Appendix D.1 show that \texttt{SMT(LRA)} mostly solves only simple instances of the problems and, on average, SMPMC can solve these instances 20x faster. Moreover, SMPMC is able to solve some very challenging problems, e.g., a benchmark where the robust minimal DT has 15 nodes.
In many cases, the underlying colored MDP has above 20K states and the number of unconstrained assignments goes beyond $10^{100}$ (see~Appendix D.3).

\paragraph{Demonstration.} 
We now describe the results for the \emph{Rocks-6-4} ME-MDP taken from \cite{DBLP:conf/atva/AndriushchenkoCJM24}. The goal of the agent is to collect four rocks on a 6x6 grid, where the environment encodes the position of the rocks. While the authors acknowledged that \enquote{there exists a single policy that is winning for all MDPs}, their algorithm that searches for unconstrained memoryless policies using a game-based abstraction and various heuristics cannot find it. 
In contrast, SMPMC finds a minimal 5-node DT that describes a robust policy, shown in \cref{fig:rockstree}.
The policy visits all grid cells, thereby ensuring it collects all rocks. It achieves this in an interesting way: in the \emph{Rocks-6-4} model, the agent slips to the left of the current movement direction with probability 10\%. The tree uses this to create a \emph{spiral} effect, visiting all cells eventually, as visualized in \cref{fig:rockgrid}.

\paragraph{Summary Q1.} The results obtained for the synthesis problems in \Ca~demonstrate that SMPMC solves a class of problems that
goes beyond the capabilities of existing tools.

\forestset{
    default preamble = {
        for tree={draw, s sep = 5 mm, edge={thick, draw=gray}, draw=gray, baseline},
        leaf/.style={
            draw, circle,
            minimum size=0pt,
            inner sep=1.5pt,
        }
    }
}

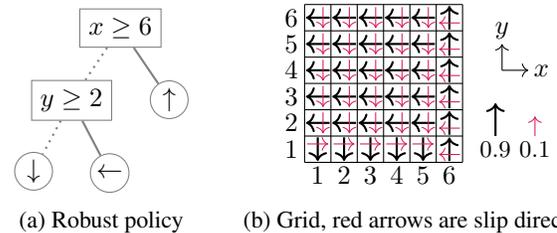
\begin{figure}
    \centering
    \begin{subfigure}[b]{0.4\linewidth}
        \centering
        \begin{forest}
            [$x \geq 6$ [$y \geq 2$, edge={dotted} [$\downarrow$, leaf, edge={dotted}] [$\leftarrow$, leaf]] [$\uparrow$, leaf]]
        \end{forest}
        \caption{Robust policy}
        \label{fig:rockstree}
    \end{subfigure}
    \begin{subfigure}[b]{0.59\linewidth}
        \centering
        \input{resources/figures/rocks}
        \caption{Grid, red arrows are slip direction}
        \label{fig:rockgrid}
    \end{subfigure}
    \caption{Minimal decision tree for \emph{Rocks-6-4}, by SMPMC.}
\end{figure}

\subsection*{Q2: Comparison on problem fragments}

\paragraph{Baselines.}
Besides the monolithic \texttt{SMT(LRA)} baseline used in Q1, more baselines are applicable to Q2.
We compare against the state-of-the-art synthesis tool PAYNT \cite{DBLP:journals/jair/AndriushchenkoCMJK25} as it solves the problems \Cb{} and \Cd{}, also on colored MDPs. PAYNT implements three strategies: (1)~abstraction-refinement (AR), (2)~counterexample-guided inductive synthesis (CEGIS) and (3)~Hybrid, combining AR and CEGIS. These strategies %
have been optimized towards the benchmarks in \Cd. \texttt{PAYNT-CEGIS} was specifically developed for problems in \Cb. We have mildly extended \texttt{\texttt{PAYNT-AR}} to support the problems in \Cc~by adding 
an inner (non-naive) 1-by-1 evaluation.
See Appendix D.3 for details and a comparison with an alternative \texttt{PAYNT-AR} extension.
Some dedicated approaches support subsets or variations of the problems defined in the classes \Cb-\Cd\, but the differences prevent a direct performance comparison. We discuss these approaches in the related work.

\paragraph{\Cb: Constrained Synthesis.}
Our benchmarks include two application domains:
1) In a set of MCs, find an MC that satisfies a given bound on the implementation cost~\cite{vcevska2021counterexample}.
2)~For a given POMDP, find a DT with $n$ nodes encoding a satisfying memoryless policy. We again consider \(n \in \{1, \ldots, 15\}\).
This benchmark set includes, e.g., the \emph{"refuel-06"} POMDP where the minimal decision tree has 11 nodes. 
 \cref{tab:results} shows that SMPMC is on par with \texttt{PAYNT-CEGIS} and the \texttt{SMT(LRA)} baseline significantly lags behind. The results in Fig.\ 9 (Appendix D.2) show that there is no clear winner in terms of runtimes.

\paragraph{\Cc: Robust Synthesis.}
These benchmarks include synthesis of robust (finite-state) policies in ME-(PO)MDPs. Additionally, we consider sets of MCs with for-all quantified variables.
\cref{tab:results} shows that SMPMC is on par with the extended version \texttt{\texttt{PAYNT-AR}}: Both methods solve the same number of benchmarks. Interestingly, there are 16 benchmarks solved only by SMPMC and 16 different benchmarks solved only by \texttt{\texttt{PAYNT-AR}}.
Results in Fig.\ 10 indicate that on the benchmarks solved by both tools, SMPMC is generally faster. \texttt{SMT(LRA)} again significantly lags behind.

\paragraph{\Cd: Plain Synthesis.}
We consider benchmarks 
including (1)~synthesis of a satisfying fixed-memory FSCs and 2) synthesis of a satisfying MC in a given set of MCs. We also include benchmarks with over 1 million states to showcase that the size of the state space on its own does not pose problems for SMPMC.
\begin{figure}
\centering
    \pgfplotsset{
        width=1\linewidth,
        height=4.53cm
    }
    \renewcommand{\pathtostuff}{experiments/simple-2025-07-29_21-15-45/}\input{experiments/simple-2025-07-29_21-15-45/quantile_processed.tex}
    \caption{Cactus Plot for \Cd. A method's line going through a point \((x, y)\) means that the method solved the \(x\)th benchmark, sorted by time, within \(y\) seconds (log-scale).}
    \label{fig:quantilefsc}
\end{figure}
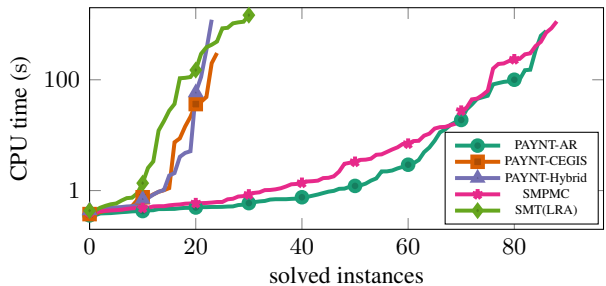
\cref{tab:results} shows that SMPMC is on par with \texttt{\texttt{PAYNT-AR}}, which is currently the fastest algorithm for these problems~\cite{DBLP:journals/jair/AndriushchenkoCMJK25}: SMPMC solves two more benchmarks, but \texttt{\texttt{PAYNT-AR}} has more unique solutions. Additionally, the results in \cref{fig:quantilefsc} show that SMPMC is competitive with \texttt{\texttt{PAYNT-AR}} also in terms of runtime. The \texttt{SMT(LRA)}, \texttt{PAYNT-CEGIS} and \texttt{PAYNT-Hybrid} lag significantly behind. 

\paragraph{Summary Q2.} The results obtained for the synthesis problems in \Cb-\Cd~demonstrate that SMPMC is highly competitive with state-of-the-art tools 
for various widely studied fragments of the problem statement. In contrast to these tools, SMPMC does not sacrifice generality and achieves efficiency without implementing any heuristics that would leverage the structure of the subproblems or additional information from model checking, which opens an interesting avenue for improving the performance of our approach.

\section{Related Work}
Closest to this work are the approaches in PAYNT, see the introduction and the experimental evaluation for details.

\paragraph{SMT and MILPs.}
This paper integrates SAT with PMC and thus stands in a tradition of approaches that integrate SAT-style reasoning with dedicated engines, such as mixed-integer (linear) programming \cite{DBLP:conf/cp/Hooker04,DBLP:journals/disopt/Achterberg07}. CDCL(T)-based approaches like ours exist for theories ranging from convex optimization \cite{DBLP:conf/hybrid/ShoukryNSSPT17} to logics on Kripke structures \cite{DBLP:conf/tableaux/EisenhoferARK23}. Planning Modulo Theories integrates planning tasks into SMT, enabling the use of theories \cite{DBLP:conf/ijcai/BofillEV17}. This paper leverages MBQI for quantifiers, whereas another quantifier instantiation strategy is the syntax-based E-matching \cite{DBLP:conf/cade/MouraB07}. MBQI has recently been integrated with syntax-based strategies
\cite{modelbasedfastenum}.

\paragraph{Controllers for (ME-)POMDPs.}
There exist several orthogonal approaches for controller synthesis in POMDPs, including various (monolithic) MILP formulations~\cite{DBLP:conf/aaai/AmatoBZ10,kumar2015history} or analysing the belief space~\cite{DBLP:conf/rss/KurniawatiHL08,DBLP:conf/uai/HoFRSL24}.
Such methods can be effectively integrated with FSC synthesis~\cite{andriushchenko23search}. Recently, a lot of focus has been given to finding robust controllers 
for multiple-environment (PO-)MDPs including sequential convex programming~\cite{DBLP:conf/aaai/Cubuktepe0JMST21}, value iteration methods~\cite{DBLP:journals/siamjo/NakaoJS21}, belief space analysis~\cite{DBLP:conf/aips/ChatterjeeCK0R20} or policy optimisation~\cite{DBLP:conf/icml/LinXD024,hmpomdps}.

\paragraph{DT policies for MDPs.}
The existing tools for finding DTs either construct a DT from a given (optimal) policy~\cite{DBLP:conf/tacas/AshokJKWWY21} or consider a formulation in which only minimality with respect to the depth of the tree is  
encoded~\cite{DBLP:conf/ijcai/VosV23,andriushchenko2025small}. The DT encoding considered in this work was inspired by the encoding using propositional formulas~\cite{DBLP:conf/ijcai/NarodytskaIPM18}. Outside of the world of MDPs, finding minimal-sized DTs from training data has been explored~\cite{DBLP:conf/aaai/StausKSS25}.

\section{Conclusion}
This paper presents SMPMC: \emph{satisfiability modulo probabilistic model checking}, which tightly integrates the power of satisfiability solvers and model checkers. Among various applications, this approach is the first effective method to find decision tree policies that robustly solve a set of MDPs.

\section*{Acknowledgements}

The authors are grateful to Clemens Eisenhofer for his help with Z3 and its user propagator API. \inlinegraphics{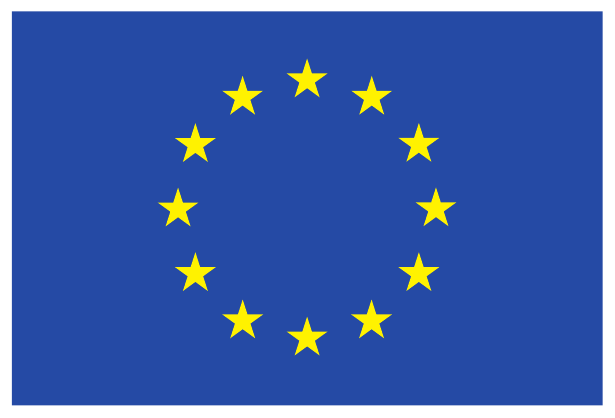} This work has been executed under the project VASSAL: ``Verification and Analysis for Safety and Security of Applications in Life'' funded by the European Union under Horizon Europe WIDERA Coordination and Support Action/Grant Agreement No. 10116002. This work has been funded by the NWO VENI Grant ProMiSe (222.147) and the Czech Science Foundation grant \mbox{GA23-06963S} (VESCAA).

\nocite{DBLP:journals/fac/ChrszonDKB18}
\nocite{DBLP:conf/atva/AndriushchenkoCJM24}
\nocite{DBLP:conf/vmcai/BjornerEK23}

\bibliography{paper}

\clearpage

\newpage

\begingroup
\appendix
\input{appendix}

\endgroup

\end{document}

%% file: resources/figures/grid.tex
\begin{tikzpicture}[x=0.6cm, y=0.6cm]

\node at (3.5,-1) {\small Environments};
\node[rotate=0] at (1,-0.15) {$e_1$};
\node[rotate=0] at (2,-0.15) {$e_2$};
\node[rotate=0] at (3,-0.15) {$e_3$};
\node[rotate=0] at (4,-0.15) {$e_4$};
\node at (5,-0.15) {$\dots$};
\node[rotate=0] at (6,-0.15) {$e_n$};

\node[rotate=90] at (-1.3,4) {\small Policies};
\node at (-0.2,1) {$\pi_1$};
\node at (-0.2,2) {$\pi_2$};
\node at (-0.2,3) {$\pi_3$};
\node at (-0.2,4) {$\pi_4$};
\node[anchor=center, yshift=0.2em] at (-0.2,5) {$\vdots$};
\node at (-0.2,6) {$\pi_m$};

\foreach \i/\row in {
  1/{c,x,c,x,x,c},
  2/{o,o,o,o,o,o},
  3/{c,c,c,c,c,c},
  4/{c,x,c,x,c,c},
  5/{c,c,c,x,c,c},
  6/{x,c,x,c,c,x},
}{
  \foreach \j [count=\x from 1] in \row {
    \ifnum\x=5
      \ifnum\i=5
        \node[anchor=center, yshift=0.2em] at (\x,\i) {$\reflectbox{$\ddots$}$};
      \else
        \node at (\x,\i) {$\cdots$};
      \fi
    \else
      \ifnum\i=5
        \node[anchor=center, yshift=0.2em] at (\x,\i) {$\vdots$};
      \else
        \if\j c
            \node[fill=green!30, minimum size=0.601cm, inner sep=0pt, draw=green!30] at (\x,\i) {\yes};
        \else\if\j x
            \node[fill=red!30, minimum size=0.601cm, inner sep=0pt, draw=red!30] at (\x,\i) {\no};
        \else
            \node[fill=gray!30, minimum size=0.601cm, inner sep=0pt, draw=gray!30] at (\x,\i) {};
            \node[
            pattern=north east lines,
            pattern color=gray!60,
            minimum size=0.601cm,
            inner sep=0pt,
            draw=none
            ] at (\x,\i) {\dontmatter};
        \fi\fi
      \fi
    \fi
  }
}

\draw[ultra thick, green!60!black, rounded corners]
  ($(-0.5,3)+(-0.4,-0.4)$) rectangle ($(6,3)+(0.4,0.4)$);

\draw (-0.5,2) -- (6.4,2);

\node[draw, thick, rounded corners, fill=white, anchor=west, align=left, inner sep=6pt] at (7.0,5) {%
  \small \makebox[\widthof{\yes}][c]{\dontmatter}:\hspace{0.5em}$(\pi_i,e_j) \notin \mathbb{V}$ \\
  \small \yes:\hspace{0.5em}$V(\pi_i,e_j) \geq \nu$\\
  \small \makebox[\widthof{\yes}][c]{\no}:\hspace{0.5em}$V(\pi_i,e_j) < \nu$
};

\node[anchor=west, align=left] at (7,3) {\small \textcolor{green!60!black}{robust policy}};

\end{tikzpicture}

%% file: resources/figures/bug.tex
\begin{tikzpicture}[x=0.8cm,y=0.8cm]

\foreach \i/\row in {
  1/{b,x,y},
  2/{g,o,r},
  3/{g,g,t},
}{
  \foreach \j [count=\x from 1] in \row {
    \if\j r
        \node[fill=yellow!30, minimum size=0.8cm, inner sep=0pt, draw=black] at (\x,\i) {};
    \else\if\j g
        \node[fill=blue!30, minimum size=0.8cm, inner sep=0pt, draw=black] at (\x,\i) {};
    \else\if\j x
        \node[fill=red!30, minimum size=0.8cm, inner sep=0pt, draw=black] at (\x,\i) {};
    \else\if\j y
        \node[fill=ForestGreen!40, minimum size=0.8cm, inner sep=0pt, draw=black] at (\x,\i) {};
    \else\if\j t
        \node[fill=black!10, minimum size=0.8cm, inner sep=0pt, draw=black] at (\x,\i) {
        };
        \node[scale=2] at (\x,\i) {\good};
    \else\if\j b
        \node[fill=yellow!30, minimum size=0.8cm, inner sep=0pt, draw=black] at (\x,\i) {
            
        };
    \begin{scope}[rotate around={-45:(\x,\i)}]
        \draw[] (\x-0.05,\i+0.25) .. controls (\x-0.2,\i+0.4) .. (\x-0.35,\i+0.45);
        \draw[] (\x+0.05,\i+0.25) .. controls (\x+0.2,\i+0.4) .. (\x+0.35,\i+0.45);
        \foreach \y in {-0.1, 0, 0.1} {
          \draw (\x+0.1,\i+\y) -- (\x+0.35,\i+\y+0.05+\y);
          \draw (\x-0.1,\i+\y) -- (\x-0.35,\i+\y+0.05+\y);
        }
        \draw[fill=ruby, draw=ruby] (\x,\i) ellipse (0.2 and 0.3);
        \fill[white] (\x-0.08,\i+0.21) circle (0.04);
        \fill[white] (\x+0.08,\i+0.21) circle (0.04);
        \fill[black] (\x-0.08,\i+0.22) circle (0.02);
        \fill[black] (\x+0.08,\i+0.22) circle (0.02);
    \end{scope}
    \else
        \node[fill=black!10, minimum size=0.8cm, inner sep=0pt, draw=black] at (\x,\i) {
            
        };
\fill[white] (\x,\i) circle (0.4);
\fill[blue!40] (\x,\i) circle (0.3); %
\fill[blue!60] (\x,\i) circle (0.1);
\foreach \angle in {60, 120, 240, 300} {
  \draw[blue!80!black, thick]
    (\x,\i)
    .. controls
    ({\x + 0.3*cos(\angle+15)}, {\i + 0.5*sin(\angle+15)}) and
    ({\x + 0.6*cos(\angle)}, {\i + 0.7*sin(\angle)})
    ..
    ({\x + 0.8*cos(\angle)}, {\i + 0.2*sin(\angle)});
}
    \fi\fi\fi\fi\fi\fi
  }
}
\end{tikzpicture}

%% file: resources/figures/robustbug.tex
\begin{tikzpicture}[x=0.8cm,y=0.8cm]

\foreach \i/\row in {
  1/{b,x,c},
  2/{g,o,r},
  3/{d,g,t},
}{
  \foreach \j [count=\x from 1] in \row {
    \if\j r
        \node[fill=yellow!30, minimum size=0.8cm, inner sep=0pt, draw=black] at (\x,\i) {};
    \else\if\j g
        \node[fill=blue!30, minimum size=0.8cm, inner sep=0pt, draw=black] at (\x,\i) {};
    \else\if\j t
        \node[fill=black!10, minimum size=0.8cm, inner sep=0pt, draw=black] at (\x,\i) {
        };
        \node[scale=2] at (\x,\i) {\good};
    \else\if\j x
        \node[fill=red!30, minimum size=0.8cm, inner sep=0pt, draw=black] at (\x,\i) {};
    \else\if\j b
        \node[fill=yellow!30, minimum size=0.8cm, inner sep=0pt, draw=black] at (\x,\i) {
        };
          \begin{scope}[rotate around={-45:(\x,\i)}]

        \draw[] (\x-0.05,\i+0.25) .. controls (\x-0.2,\i+0.4) .. (\x-0.35,\i+0.45);
        \draw[] (\x+0.05,\i+0.25) .. controls (\x+0.2,\i+0.4) .. (\x+0.35,\i+0.45);
        \foreach \y in {-0.1, 0, 0.1} {
          \draw (\x+0.1,\i+\y) -- (\x+0.35,\i+\y+0.05+\y);
          \draw (\x-0.1,\i+\y) -- (\x-0.35,\i+\y+0.05+\y);
        }
        \draw[fill=ruby, draw=ruby] (\x,\i) ellipse (0.2 and 0.3);
        \fill[white] (\x-0.08,\i+0.21) circle (0.04);
        \fill[white] (\x+0.08,\i+0.21) circle (0.04);
        \fill[black] (\x-0.08,\i+0.22) circle (0.02);
        \fill[black] (\x+0.08,\i+0.22) circle (0.02);
        \end{scope}
    \else\if\j c
        \node[fill=ForestGreen!40, minimum size=0.8cm, inner sep=0pt, draw=black] at (\x,\i) {};
          \begin{scope}[rotate around={30:(\x,\i)}]

        \draw[] (\x-0.05,\i+0.25) .. controls (\x-0.2,\i+0.4) .. (\x-0.35,\i+0.45);
        \draw[] (\x+0.05,\i+0.25) .. controls (\x+0.2,\i+0.4) .. (\x+0.35,\i+0.45);
        \foreach \y in {-0.1, 0, 0.1} {
          \draw (\x+0.1,\i+\y) -- (\x+0.35,\i+\y+0.05+\y);
          \draw (\x-0.1,\i+\y) -- (\x-0.35,\i+\y+0.05+\y);
        }
        \draw[fill=ruby, draw=ruby] (\x,\i) ellipse (0.2 and 0.3);
        \fill[white] (\x-0.08,\i+0.21) circle (0.04);
        \fill[white] (\x+0.08,\i+0.21) circle (0.04);
        \fill[black] (\x-0.08,\i+0.22) circle (0.02);
        \fill[black] (\x+0.08,\i+0.22) circle (0.02);
        \end{scope}
    \else\if\j d
        \node[fill=blue!30, minimum size=0.8cm, inner sep=0pt, draw=black] at (\x,\i) {};
          \begin{scope}[rotate around={-60:(\x,\i)}]

        \draw[] (\x-0.05,\i+0.25) .. controls (\x-0.2,\i+0.4) .. (\x-0.35,\i+0.45);
        \draw[] (\x+0.05,\i+0.25) .. controls (\x+0.2,\i+0.4) .. (\x+0.35,\i+0.45);
        \foreach \y in {-0.1, 0, 0.1} {
          \draw (\x+0.1,\i+\y) -- (\x+0.35,\i+\y+0.05+\y);
          \draw (\x-0.1,\i+\y) -- (\x-0.35,\i+\y+0.05+\y);
        }
        \draw[fill=ruby, draw=ruby] (\x,\i) ellipse (0.2 and 0.3);
        \fill[white] (\x-0.08,\i+0.21) circle (0.04);
        \fill[white] (\x+0.08,\i+0.21) circle (0.04);
        \fill[black] (\x-0.08,\i+0.22) circle (0.02);
        \fill[black] (\x+0.08,\i+0.22) circle (0.02);
        \end{scope}
    \else
            \node[fill=black!10, minimum size=0.8cm, inner sep=0pt, draw=black] at (\x,\i) {
            
        };
\fill[white] (\x,\i) circle (0.4);
\fill[blue!40] (\x,\i) circle (0.3); %
\fill[blue!60] (\x,\i) circle (0.1);
\foreach \angle in {60, 120, 240, 300} {
  \draw[blue!80!black, thick]
    (\x,\i)
    .. controls
    ({\x + 0.3*cos(\angle+15)}, {\i + 0.5*sin(\angle+15)}) and
    ({\x + 0.6*cos(\angle)}, {\i + 0.7*sin(\angle)})
    ..
    ({\x + 0.8*cos(\angle)}, {\i + 0.2*sin(\angle)});
}
    \fi\fi\fi\fi\fi\fi\fi
  }
}
\end{tikzpicture}

%% file: resources/figures/rocks.tex
\begin{tikzpicture}[x=0.35cm, y=0.35cm]

\foreach \x in {0.5,1.5,...,6.5} {
  \draw (\x,0.5) -- (\x,6.5);
}
\foreach \y in {0.5,1.5,...,6.5} {
  \draw (0.5,\y) -- (6.5,\y);
}

\node[rotate=0] at (1,0) {$1$};
\node[rotate=0] at (2,0) {$2$};
\node[rotate=0] at (3,0) {$3$};
\node[rotate=0] at (4,0) {$4$};
\node[rotate=0] at (5,0) {$5$};
\node[rotate=0] at (6,0) {$6$};

\node at (-0,1) {$1$};
\node at (-0,2) {$2$};
\node at (-0,3) {$3$};
\node at (-0,4) {$4$};
\node at (-0,5) {$5$};
\node at (-0,6) {$6$};

\draw[->, black] (8,4) -- (8,5);
\node at (8,5.5) {$y$};
\draw[->, black] (8,4) -- (9,4);
\node at (9.5,4) {$x$};

\draw[->, thick, black] (7.75,1.5) -- (7.75,2.75);
\node at (7.75,1) {\footnotesize$0.9$};
\draw[->, ruby] (9.25,1.5) -- (9.25,2.25);
\node at (9.25,1) {\footnotesize$0.1$};

\foreach \i/\row in {
  1/{d,d,d,d,d,u},
  2/{l,l,l,l,l,u},
  3/{l,l,l,l,l,u},
  4/{l,l,l,l,l,u},
  5/{l,l,l,l,l,u},
  6/{l,l,l,l,l,u},
}{
  \foreach \j [count=\x from 1] in \row {
        \coordinate (C) at (\x,\i);
      \if\j u
      \coordinate (A) at (0,0.45);
      \coordinate (A2) at (0,0.2);
      \coordinate (L) at (-0.4,0);
    \else\if\j d
      \coordinate (A) at (0,-0.45);
      \coordinate (A2) at (0,-0.2);
      \coordinate (L) at (0.4,0);
    \else\if\j r
      \coordinate (A) at (0.45,0);
      \coordinate (A2) at (0.2,0);
      \coordinate (L) at (0,0.4);
    \else
      \coordinate (A) at (-0.45,0);
      \coordinate (A2) at (-0.2,0);
      \coordinate (L) at (0,-0.4);
    \fi\fi\fi

    \draw[->, thick, black] ($(C)-(A)$) -- ($(C)+(A)$);
    \draw[->, ruby] ($(C)-(L)-(A2)$) -- ($(C)+(L)-(A2)$);

  }
}
\end{tikzpicture}

%% file: experiments/simple-2025-07-29_21-15-45/quantile_processed.tex
\begin{tikzpicture}
\begin{semilogyaxis}[
        /pgfplots/table/header=false,
        xlabel={\small solved instances},
        ylabel={\small CPU time (\second)},
        label style={font=\small},
        xlabel style={at={(ticklabel cs:0.5,-2.0)},anchor=near ticklabel},
        ylabel style={at={(ticklabel cs:0.5,-2.0)},anchor=near ticklabel},
        xmin=0,
        ymin=0.2,
        ymax=2000,
        mark repeat=500,
        cycle multiindex* list={
            Dark2 \nextlist
            mark list
        },
        mark repeat=10,    %
        mark size=2pt,     %
        legend style={nodes={scale=0.5, transform shape}},
        legend entries={PAYNT-AR,PAYNT-CEGIS,PAYNT-Hybrid,SMPMC,SMT(LRA),},
        every axis legend/.append style={at={(1,0)}, anchor=south east, outer xsep=2pt, outer ysep=2pt,},
        xticklabel style={font=\small},
        yticklabel style={font=\small},
        ]
        \foreach \tool in {comparison_simple.2025-07-29_21-15-45.results.PAYNT-AR.xml.bz2.quantile.csv,comparison_simple.2025-07-29_21-15-45.results.PAYNT-CEGIS.xml.bz2.quantile.csv,comparison_simple.2025-07-29_21-15-45.results.PAYNT-Hybrid.xml.bz2.quantile.csv,comparison_simple.2025-07-29_21-15-45.results.SMPMC.xml.bz2.quantile.csv,comparison_simple.2025-07-29_21-15-45.results.SMTLRA.xml.bz2.quantile.csv} {
            \addplot+ [line width=1.5pt] table[y index=5] {\pathtostuff\tool};
        }
\end{semilogyaxis}
\end{tikzpicture}

%% file: appendix.tex
\clearpage
\section{Full Algorithm}
\label{appendix:algorithm}

The full algorithm for our theory solver is in \cref{alg:mdpcheck}. It attempts to prove that a set of literals \(K\) is inconsistent using \cref{prop:inconsistent} and if it is inconsistent, computes \(C(K)\) as defined in \cref{theorem:smallerconflicts} and pushes this back to the SMT solver.

\paragraph{Integration into Z3.}
Z3 calls the user propagator's \textsf{Fixed} callback when assigning some value. It calls the \textsf{Final} callback after deciding on a single assignment. In both cases, the user propagator checks whether a conflict arises using \cref{prop:inconsistent}.
In problems where Z3 uses quantifier instantiation, Z3 calls the \textsf{Fresh} callbacks to create new solver contexts and the \textsf{Created} callback when creating new parametrizations of the \(\viable\) function. The user propagator keeps track of these events to be able to report conflicts in the new context~\cite{DBLP:conf/vmcai/BjornerEK23}.

\paragraph{Implementation details.}

Our implementation of the theory solver corresponds exactly to \cref{alg:mdpcheck} except for two small optimizations:
\begin{itemize}
    \item In the beginning of the solve, we do not compute or report conflicts until the first assignment is visited by the SMT solver. This way, we can give the SMT solver a chance to \enquote{guess} a satisfying assignment immediately, which can sometimes improve runtime.
    \item As SMPMC spends most of its time in probabilistic model checking, we employ a cache of input-output pairs for the theory solver that returns an answer if it was already computed. We use incremental datastructures to amortize computational cost from solving similar problems successively.
\end{itemize}

The implementation uses Storm 1.10.1 and Z3 4.13.3.0.

\begin{algorithm}
\begin{algorithmic}
    \Procedure{CheckMDPs}{$\C, \nu, K$}
        \For{all literals $k \in K$}
            \If{$v = \viable(p_1, p_2, \ldots, p_n)$}
                \State $({\sim}, d) \leftarrow ({<}, {\max})$
            \ElsIf {$v = \lnot \viable(p_1, p_2, \ldots, p_n)$}
                \State $({\sim}, d) \leftarrow ({\geq}, {\min})$
            \Else
                \State \textbf{continue}
            \EndIf
            \State $\eta \leftarrow \{\theta \mid \theta \text{ satisfies } K \setminus \{v\}\}\}$
            \If{$V^d(\C[\eta]) \sim \nu$} \Comment{(\cref{prop:inconsistent})}
                \State $C \leftarrow C(K)$ \Comment{Get conflict set (\cref{theorem:smallerconflicts})}
                \State \Call{PushConflict}{$C \cup \{v\}$}
            \EndIf \Comment{No conflict occured, do nothing}
        \EndFor
    \EndProcedure
\end{algorithmic}
\caption{Theory MDP Check.}
\label{alg:mdpcheck}
\end{algorithm}

\section{Proofs and Extended Example}

\subsection{Extended Example for Theorems}
\label{app:extendedexample}

\begin{figure}
        \centering%
        \begin{tikzpicture}[every state/.style={inner sep=0pt, minimum size=16pt}, node distance=1cm]
            \node[state, fill=yellow!30] (s00) {$s_{0,0}$};
            \node[state, fill=red!30] (s01) [right of=s00] {$s_{0,1}$};
            \node[state, fill=ForestGreen!40] (s02) [right of=s01] {$s_{0,2}$};
            
            \node[state, fill=blue!30] (s10) [above of=s00] {$s_{1,0}$};
            \node[state] (s11) [right of=s10] {$s_{1,1}$};
            \node[state, fill=yellow!30] (s12) [right of=s11] {$s_{1,2}$};
            
            \node[state, fill=blue!30] (s20) [above of=s10] {$s_{2,0}$};
            \node[state, fill=blue!30] (s21) [right of=s20] {$s_{2,1}$};
            \node[state] (s22) [right of=s21] {$\good$};
            
            \path[->] (s00) edge node[left] {$1$} (s10);
            
            \path[->] (s20) edge[loop left] node[left] {$1$} (s21);
            \path[->] (s10) edge[loop left] node[left] {$1$} (s11);
            \path[->] (s21) edge node[above] {$1$} (s20);
            \path[->] (s01) edge[loop below] node[below] {$1$} (s01);
            
            \path[->] (s22) edge[loop right] node[right] {$1$} (s02);
            \path[->] (s12) edge node[right] {$1$} (s22);
            
            \path[->] (s11) edge node[left] {$\nicefrac{1}{4}$} (s21);
            \path[->] (s11) edge node[right] {$\nicefrac{1}{4}$} (s01);
            \path[->] (s11) edge node[above] {$\nicefrac{1}{4}$} (s12);
            \path[->] (s11) edge node[below] {$\nicefrac{1}{4}$} (s10);
            
            \path[->] (s02) edge[loop right] node[right] {$1$} (s22);
            \path[->] (s02) edge[loop below] node[below] {$1$} (s22);
            \path[->] (s02) edge node[right] {$1$} (s12);
            \path[->] (s02) edge node[below] {$1$} (s01);
    \end{tikzpicture}
    \caption{\(\C[\eta]\) in extended example.}
    \label{fig:bugmdp}
\end{figure}
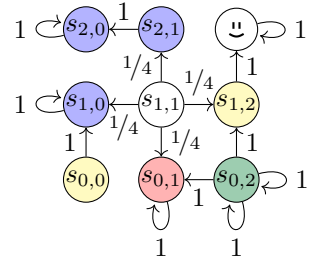

Again, we consider the situation in \cref{fig:bug1} and have $K =\{\viable(d_r,d_g,d_b,d_y), d_y\eq 0,d_r\eq 2,d_b\eq 1\}$. We invoke \cref{prop:inconsistent}:

\inconsistent*

Here, we have:
\begin{itemize}
    \item \(v = \viable(d_r,d_g,d_b,d_y)\),
    \item \(\eta = \{\theta \mid \theta(d_y)\eq 0 , \theta(d_r)\eq 2, \theta(d_b)\eq 1\}\).
\end{itemize}

The MDP \(\C[\eta]\), depicted in \cref{fig:bugmdp}, has a nondeterministic choice at the green state and fixed choices at all other states. Indeed, we have \(V^{\max}(\C[\eta]) = 0 < 1\), as no matter which choice is taken at the green state, the bug will not reach the target. Now we discuss computing the smaller conflict \(C(K)\). Intuitively, the choice of direction on red tiles is also irrelevant, as the bug never reach such a tile. We invoke \cref{theorem:smallerconflicts}:

\smallerconflicts*

Here, we have \(s_I = s_{0,0}\), so the reachable states are \(S = \{s_{0, 0}, s_{1, 0}\}\). Indeed, the reachable colorings are only dependent on parameters \(d_y\) and \(d_b\) (we provide a detailed definition of this in \cref{app:proofs}). This means that we can compute this smaller conflict: \(C(K) = \{\viable(d_r,d_g,d_b,d_y), d_y\eq 0,d_b\eq 1\}\).

\subsection{Proofs}
\label{app:proofs}

\validtheorem*

\begin{proof}
\enquote{\(\Rightarrow\)}: Suppose \(\theta_X\) is such a satisfying assignment. Then the following sentence is valid in \(\mathbb{T}_\C\):
\begin{align*}
    &\forall y_1 \ldots  \forall y_m:\tau_{\mathbb{V}}(x_1/\theta_X(x_1), \ldots, x_n/\theta_X(x_n), y_1, \ldots, y_m)\\ &\rightarrow \; \viable(x_1/\theta_X(x_1), \ldots, x_n/\theta_X(x_n), y_1, \ldots, y_m).
\end{align*}
For any given \(\theta_Y \in \mathbb{Z}^Y\) with \(\tau_{\mathbb{V}}(\theta_X \cup \theta_Y)\), we thus have \(\viable(\theta_X(x_1), \ldots, \theta_X(x_n), \theta_Y(y_1), \ldots, \theta_Y(y_m)) \in \mathbb{T}_\C\), and by definition of \(\mathbb{T}_\C\), we have \(V(\C[\theta_X \cup \theta_Y]) \geq \nu\).

\enquote{\(\Leftarrow\)}: Analogous.

We note that for conciseness, the assumption \(\theta_X \in \mathbb{V}_X\) rules out cases where \(\theta_X \in \mathbb{Z}^X \setminus \mathbb{V}_X\). We can also state the stronger theorem for all \(\theta_X \in \mathbb{Z}^X\) by appending this constraint to \(\Phi\) to yield the following statement:
\begin{align*}
    \Phi' := \; &\exists x_1 \ldots \exists x_n : \; \\
    &\exists y_1 \ldots  \exists y_m: \tau_{\mathbb{V}}(x_1, \ldots, x_n, y_1, \ldots, y_m) \land \\
    &(\forall y_1 \ldots  \forall y_m:\tau_{\mathbb{V}}(x_1, \ldots, x_n, y_1, \ldots, y_m)\\ &\rightarrow \; \viable(x_1, \ldots, x_n, y_1, \ldots, y_m)).
\end{align*}
\end{proof}

\inconsistent*

\begin{proof}
    Suppose \(\viable(p_1, \ldots, p_n) \in K\), let \(\eta\) be defined as above and let \(\C = (\M, \mathbb{V}, \kappa)\) be the colored MDP with \(V^{\max}(\C[\eta]) < \nu\). We will prove that \(K\) is a \conflict{}.
    We first define a new colored MDP \(\C_\eta:=(\C[\eta], \eta, \kappa)\). Intuitively, \(\C_\eta\) is \(\C\) constrained to the new parameter space \(\eta\), so only parameter instantiations admissible by \(\eta\) are represented in \(\C_\eta\). Note that \(\C[\eta'] = \C_\eta[\eta']\) for all \(\eta' \subseteq \eta\).
    We can now state
    \begin{align*}
        & \max_{\theta \in \eta} V(\C[\theta])
        = \max_{\theta \in \eta} V(\C_\eta[\theta]) \\
        \overset{(1)}{\leq} & V^{\max}(\C_\eta[\eta])
        = V^{\max}(\C[\eta])
        \overset{(2)}{<} \nu.
    \end{align*}
    Step \((1)\) is established in  \cite[Theorem 1]{DBLP:journals/jair/AndriushchenkoCMJK25} and step \((2)\) is true by assumption. We thus know that for all \(\theta \in \eta\), we have \(V(\C[\theta]) < \nu\).
    
    Let \(\theta \in \mathbb{V}\). We have two cases:
    \begin{itemize}
    \item
    If \(\theta \notin \eta\), then \(\theta\) does not satisfy \(K \setminus \{v\}\) by definition of \(\eta\), and thus it does not satisfy \(K\).
    \item
    If \(\theta \in \eta\), then we have \(V(\C[\theta]) < \nu\) as shown above, which implies \(\lnot\viable(\theta(p_1), \ldots \theta(p_n)) \in \mathbb{T}_\C\) by definition of \(\mathbb{T}_\C\). As \(\viable(p_1, \ldots, p_n) \in K\), \(\theta\) does not satisfy \(\mathbb{T}_\C \cup K\).
    \end{itemize}
    Because no assignment \(\theta \in \mathbb{V}\) satisfies \(\mathbb{T}_\C \cup K\), \(K\) is a \conflict{}.
    The symmetric case is analogous.
\end{proof}

\smallerconflicts*

We first define (in)dependence and the restricted partial assignment \(C(K)\).

\begin{definition}[Coloring Dependent on Parameter]
We say that a coloring \(\kappa(s, \alpha)\) is \emph{dependent on a parameter} \(p_i \in X \cup Y\) if there exist two assignments \(\theta, \theta' \in \mathbb{V}\), and a \(z \in \mathbb{Z}\) such that:
\begin{itemize}
    \item \(z \neq \theta(p_i)\),
    \item \(\theta'(p_j) = \begin{cases}\theta(p_j) &\text{ if } i \neq j \\ z &\text{ if } i = j\end{cases}\),
    \item \(\theta \in \kappa(s,\alpha)\) and \(\theta' \notin \kappa(s, \alpha)\).
\end{itemize}
Conversely, a coloring \(\kappa(s, \alpha)\) is \emph{independent} of \(p_i\) if it is not dependent on \(p_i\).
\end{definition}

Suppose we have the colored MDP \(\C\), the set of assignments \(\eta\), and \(\C[\eta] = (S, s_I, Act, \mathcal{P}, R)\). Then a \(\C[\eta]\) is dependent on \(p_i\) if some coloring \(\kappa(s, \alpha)\), where \(s\) is reachable from \(s_I\), is dependent on \(p_i\).

Now we can define \(C(K)\):
\begin{align*}
C(K) := \;& \{v\}\; \cup \\ & \{(p_i\eq z) \in K \mid \C[\eta] \text{ is dependent on } p_i\}.
\end{align*}

A short note on computing (in)dependence. In our implementation, we describe the set \(\kappa(s, \alpha)\) using a partial assignment. Here, a parameter is dependent if it occurs in this partial assignment, making dependence easy to check by computing the set of reachable states and checking which parameters occur in the reachable partial assignments.

\begin{proof}
We show that the proof for \cref{prop:inconsistent} goes through if we replace \(K\) by \(C(K)\). It suffices to show that for \[
    \eta' = \{\theta \in \mathbb{V} \mid \theta \text{ satisfies } C(K) \setminus \{v\} \},
\] we have \(V(\C_{\eta'}[\theta]) = V(\C_\eta[\theta])\) for all \(\theta \in \eta'\). Let \(P := \{p_i  \mid \C_\eta \text{ is independent of } p_i\}\). Because (1) the elements of \(\eta\) and \(\eta'\) only differ in parameters in \(P\) and (2) \(\C_\eta\) or and \(\C_\eta'\) do not depend on any parameters in \(P\), the values of these parameters can never influence the values \(V(\C_{\eta}[\theta])\) or \(V(\C_\eta[\theta])\). Thus, the statement follows.
\end{proof}

\section{Supplementary Definitions}

\subsection{POMDPs and Colored MDPs}

We will start by defining POMDPs.

\begin{definition}[Partially Observable MDP]
    A \emph{partially observable MDP (POMDP)} is a tuple \(\mathcal{D} = (\M, Z, O)\) with an MDP \(\M\), a finite set \(Z\) of observations and an observation function \(O : S \rightarrow Z\).
\end{definition}

In stark contrast to MDPs, where an optimal policy can be found efficiently given the model, policy synthesis in POMDPs is undecidable. We restrict ourselves to synthesizing \emph{deterministic fixed-memory finite-state controllers (FSCs)}, which is a decidable, NP-complete
fragment of policy synthesis. By a preprocessing step, we can even assume that our controllers are memoryless by using an unrolling of the POMDP. The idea is to unroll the POMDP into a new POMDP with \(|S| \cdot n\) states and \(|Z| \cdot n\) observations, on which a memoryless policy corresponds to an FSC on the original POMDP. The memoryless policies are functions \(\pi: Z \rightarrow Act\) on this unrolled POMDP.

\begin{problem}[Memoryless POMDP Policy]
    Given POMDP \(\mathcal{D}\) and specification \(\varphi\), is there a memoryless policy $\pi : Z \rightarrow Act$ s.t. \(\mathcal{D}[\pi] \vDash \varphi\)?
\end{problem}
We use the conversion from POMDPs to Colored MDPs from \cite{DBLP:journals/jair/AndriushchenkoCMJK25}. Here, instantiations of the colored MDP correspond to memoryless policies of the POMDP.

\subsection{Sets of POMDPs as Colored MDPs}

We now ask for a robust policy on a set of POMDPs.

\begin{problem}[Robust Memoryless POMDP Policy]
    Given a set of POMDPs \(\mathfrak{D}\) that share states $S$, actions $Act$ and observations $Z$, and specification \(\varphi\), is there a memoryless policy $x : Z \rightarrow Act$ s.t. \(\D[x] \vDash \varphi\) for all \(\D \in \mathfrak{D}\)?
\end{problem}

This can again be generalized to synthesizing FSCs by an unrolling. To translate this problem statement into our problem statement, we define a single colored MDP that captures a set of POMDPs. This notion is very similar to feature MDPs \cite{DBLP:journals/fac/ChrszonDKB18} and quotient MDPs for a set of MDPs \cite{DBLP:conf/atva/AndriushchenkoCJM24}. As in these concepts, we assume that the POMDPs only differ in the transition functions, which does not lose us generality, and often gives us a compact representation of the set.

\begin{definition}[Colored MDP from set of POMDPs]
    \label{def:coloredmdpfrompomdp}
    Given a set of POMDPs \(\mathfrak{D}\) that share \(S, s_I, Z, O, Act\) and differ in the transition functions \(\mathcal{P}_\D\) for \(\D \in \mathfrak{D}\). Suppose we have a function \(s: \mathbb{V}_Y \rightarrow \mathfrak{D}\) that selects a POMDP based on an assignment of the parameters in \(Y\).
    Let \(Act = \{\alpha_1, \ldots, \alpha_k\}\).
    The resulting colored MDP \(\C = (\M, \mathbb{V}, \kappa)\) has the MDP \(\M = (S, s_I, \mathfrak{D} \times Act, \mathcal{P})\) with \(\mathcal{P}(s, (\D, \alpha)) = \mathcal{P}_\D(s, \alpha)\). The parameters are \(V = X \cup Y\) for \(X = Z\) and the coloring is \(\kappa(s, (\D, \alpha_i)) = \{\theta \mid s(\theta_Y) = \D \text{ and } \theta(O(s, \alpha_i)) = i\}\).
\end{definition}

Note that synthesizing a robust policy for a set of MDPs is as difficult as the problem for a set of POMDPs.

\begin{proposition}
    \label{prop:coloredmdp}
    An assignment \(\theta\) of a colored MDP \(\C\) as above corresponds to the memoryless policy \(\pi: Z \rightarrow Act\) of POMDP \(u(s(\theta_Y))\) with \(\theta(o) = \pi(o)\) for \(o \in Z\).
\end{proposition}

\section{Detailed Experimental Results}

\subsection{\Ca: Constrained Robust Synthesis}
\label{appendix:c1}

\paragraph{Experimental results.}

\begin{table*}
\small\centering
\begin{tabular}{lrrrrrrrr}
\toprule
Benchmark & 1 & 3 & 5 & 7 & 9 & 11 & 13 & 15 \\
\midrule
avoid-8-2-bf & \no{} 9.94 & \yes{} 1191.00 & NR & NR & NR & NR & NR & NR \\
catch-5 & \no{} 1.27 & \no{} 3.36 & \no{} 32.54 & \no{} 749.28 & NR & NR & NR & NR \\
dpm-switch-q10-big-bf & \no{} 7.13 & \no{} 407.45 & NR & NR & NR & NR & NR & NR \\
obstacles-10-6-skip-easy-bf & \no{} 4.10 & \yes{} 5.63 & \yes{} 13.95 & \yes{} 183.15 & \yes{} 71.23 & \yes{} 66.34 & \yes{} 46.81 & \yes{} 103.96 \\
obstacles-demo & \no{} 2.47 & \no{} 5.34 & \yes{} 78.29 & \yes{} 10.87 & \yes{} 88.34 & \yes{} 709.52 & \yes{} 1266.36 & \yes{} 708.77 \\
pacman-6 & \no{} 5.94 & \yes{} 12.54 & \yes{} 35.45 & \yes{} 55.47 & \yes{} 65.45 & \yes{} 184.95 & \yes{} 694.58 & \yes{} 839.51 \\
refuel-04 & \no{} 0.49 & \no{} 0.59 & \no{} 1.22 & \yes{} 1.61 & \yes{} 1.86 & \yes{} 6.04 & \yes{} 3.74 & \yes{} 5.38 \\
\textbf{rocks-6-4} & \no{} 1.23 & \no{} 16.43 & \yes{} 249.33 & \yes{} 216.27 & \yes{} 1134.56 & NR & \yes{} 1192.25 & NR \\
rocks-6-4-signal & \no{} 1.27 & \no{} 10.52 & \yes{} 627.91 & NR & NR & NR & NR & NR \\
rover-1000-bf & \yes{} 5.74 & \yes{} 12.65 & \yes{} 9.19 & \yes{} 11.90 & \yes{} 14.57 & \yes{} 16.68 & \yes{} 28.08 & \yes{} 23.11 \\
uav-operator-roz-workload-bf & \no{} 16.07 & \no{} 58.41 & NR & NR & NR & NR & NR & NR \\
uav-roz & \no{} 17.67 & \no{} 103.98 & NR & NR & NR & NR & NR & NR \\
mem\_2\_obstacles-2-3 & \no{} 0.49 & \no{} 0.74 & \no{} 0.84 & \no{} 1.05 & \no{} 1.22 & \no{} 3.07 & \no{} 10.79 & \yes{} 4.53 \\
mem\_3\_obstacles-3-3 & \no{} 0.64 & \no{} 0.48 & \no{} 1.15 & \no{} 0.91 & \no{} 1.73 & \no{} 3.06 & \no{} 30.30 & \no{} 115.65 \\
\bottomrule
\end{tabular}
\caption{SMPMC Results for \Ca.}
\label{tab:C1:smpmc}
\end{table*}

\begin{table*}
\small\centering
\begin{tabular}{lrrrrrrrr}
\toprule
Benchmark & 1 & 3 & 5 & 7 & 9 & 11 & 13 & 15 \\
\midrule
avoid-8-2-bf & NR & NR & NR & NR & NR & NR & NR & NR \\
catch-5 & NR & NR & NR & NR & NR & NR & NR & NR \\
dpm-switch-q10-big-bf & NR & NR & NR & NR & NR & NR & NR & NR \\
obstacles-10-6-skip-easy-bf & \no{} 47.47 & NR & NR & NR & NR & NR & NR & NR \\
obstacles-demo & \no{} 278.48 & NR & NR & NR & NR & NR & NR & NR \\
pacman-6 & NR & NR & NR & NR & NR & NR & NR & NR \\
refuel-04 & \no{} 1.31 & \no{} 1.80 & \no{} 4.54 & \yes{} 213.28 & \yes{} 422.93 & \yes{} 252.03 & \yes{} 955.92 & \yes{} 628.01 \\
\textbf{rocks-6-4} & NR & NR & NR & NR & NR & NR & NR & NR \\
rocks-6-4-signal & NR & NR & NR & NR & NR & NR & NR & NR \\
rover-1000-bf & NR & NR & NR & NR & NR & NR & NR & NR \\
uav-operator-roz-workload-bf & NR & NR & NR & NR & NR & NR & NR & NR \\
uav-roz & NR & NR & NR & NR & NR & NR & NR & NR \\
mem\_2\_obstacles-2-3 & \no{} 1.50 & \no{} 1.55 & \no{} 1.62 & \no{} 1.82 & \no{} 2.28 & \no{} 3.85 & NR & NR \\
mem\_3\_obstacles-3-3 & \no{} 6.94 & \no{} 3.62 & \no{} 3.67 & \no{} 3.97 & \no{} 4.40 & \no{} 10.20 & \no{} 199.16 & \no{} 526.88 \\
\bottomrule
\end{tabular}
\caption{SMT(LRA) Results for \Ca.}
\label{tab:C1:SMTLRA}
\end{table*}

We show the complete results in \cref{tab:C1:smpmc} and \cref{tab:C1:SMTLRA}.

\subsubsection{Decision tree encoding.}
\label{app:c1enc}

We encode the synthesis of a decision tree with a predefined number of nodes ($N$) as an SMT problem.
For each node $i \in \{0, \dots, N-1\}$, we define a set of solver variables:
\begin{itemize}
    \item A boolean variable $\text{is\_leaf}_i$ determines whether node \(i\) is a leaf or a decision node.
    \item An integer $\text{prop\_index}_i$ selects the input feature that node \(i\) evaluates.
    \item An integer $\text{const}_i$ represents the threshold for comparison in a decision node, or as the output action label in a leaf node.
    \item Integers $\text{left}_i$ and $\text{right}_i$ store the indices of the child nodes.
\end{itemize}
We add constraints to enforce a valid tree structure. This includes ensuring that the number of leaf nodes is $(N+1)/2$ and that every node (except the root) has exactly one parent. For each node, we introduce an uninterpreted function that describes the decision tree's behavior from that node. For a given input, a decision node's function evaluates to its child's function based on the comparison $(\theta(\text{prop\_index}_i) \geq \text{const}_i)$. A leaf node's function simply returns its action label $\text{const}_i$. Finally, we assert that the policy is described by the top node's uninterpreted function.

\subsection{\Cb: Constrained Synthesis}
\label{appendix:c2}

\paragraph{Experimental results.}
The cost-bounded benchmarks can be seen in \cref{tab:C2:cost}. The decision tree benchmarks can be seen in \cref{tab:C2:SMPMC}, \cref{tab:C2:CEGIS}, and \cref{tab:C2:SMTLRA}.

\begin{table*}
\centering
\begin{tabular}{lrrr}
\toprule
Benchmark & PAYNT-CEGIS & SMPMC & SMT(LRA) \\
\midrule
sat\_dpm & \textbf{3.57} & $10.44$ & NR \\
sat\_grid-10-sl-4fsc & \textbf{0.59} & $1.12$ & NR \\
sat\_grid-meet-sl-2fsc & $33.07$ & \textbf{19.49} & NR \\
sat\_maze & $133.89$ & \textbf{0.47} & NR \\
unsat\_dpm & \textbf{5.55} & $8.34$ & NR \\
unsat\_grid-10-sl-4fsc & \textbf{0.77} & $1.43$ & NR \\
unsat\_grid-meet-sl-2fsc & \textbf{47.61} & $83.68$ & NR \\
unsat\_maze & NR & \textbf{1.48} & NR \\
\bottomrule
\end{tabular}
\caption{Cost-Bounded Benchmarks in \Cb{}.}
\label{tab:C2:cost}
\end{table*}

\begin{table*}
\small\centering
\begin{tabular}{lrrrrrrrr}
\toprule
Benchmark & 1 & 3 & 5 & 7 & 9 & 11 & 13 & 15 \\
\midrule
drone-4-1 & \no{} 0.77 & \no{} 4.07 & \no{} 417.78 & NR & NR & NR & NR & NR \\
drone-4-2 & \no{} 1.68 & \no{} 10.65 & NR & NR & NR & NR & NR & NR \\
drone-8-2 & \no{} 6.29 & \no{} 118.38 & NR & NR & NR & NR & NR & NR \\
grid-avoid-4-0.1 & \yes{} 0.43 & \yes{} 0.48 & \yes{} 0.44 & \yes{} 0.42 & \yes{} 0.42 & \yes{} 0.42 & \yes{} 0.38 & \yes{} 0.39 \\
maze-alex & \no{} 0.40 & \no{} 0.79 & \no{} 2.06 & \yes{} 5.56 & \yes{} 16.34 & \yes{} 14.51 & \yes{} 1.01 & \yes{} 10.27 \\
nrp-8 & \yes{} 0.42 & \yes{} 0.57 & \yes{} 0.60 & \yes{} 0.77 & \yes{} 0.59 & \yes{} 0.88 & \yes{} 0.89 & \yes{} 1.06 \\
refuel-06 & \no{} 0.49 & \no{} 0.53 & \no{} 1.37 & \no{} 12.32 & \no{} 216.05 & \yes{} 48.91 & \yes{} 64.42 & \yes{} 256.61 \\
rocks-12 & \no{} 5.80 & \no{} 10.79 & \yes{} 45.32 & \yes{} 65.66 & \yes{} 152.73 & \yes{} 381.47 & \yes{} 421.39 & \yes{} 1104.39 \\
\bottomrule
\end{tabular}
\caption{PAYNT-CEGIS Results for decision tree synthesis in \Cb.}
\label{tab:C2:CEGIS}
\end{table*}

\begin{table*}
\small\centering
\begin{tabular}{lrrrrrrrr}
\toprule
Benchmark & 1 & 3 & 5 & 7 & 9 & 11 & 13 & 15 \\
\midrule
drone-4-1 & \no{} 0.81 & \no{} 14.17 & \no{} 418.56 & NR & NR & NR & NR & NR \\
drone-4-2 & \no{} 0.98 & \no{} 603.52 & NR & NR & NR & NR & NR & NR \\
drone-8-2 & \no{} 3.22 & NR & NR & NR & NR & NR & NR & NR \\
grid-avoid-4-0.1 & \yes{} 0.62 & \yes{} 0.46 & \yes{} 0.66 & \yes{} 0.53 & \yes{} 0.53 & \yes{} 0.59 & \yes{} 0.65 & \yes{} 0.59 \\
maze-alex & \no{} 0.62 & \no{} 0.66 & \no{} 0.64 & \yes{} 0.95 & \yes{} 1.90 & \yes{} 2.10 & \yes{} 1.27 & \yes{} 1.61 \\
nrp-8 & \yes{} 0.56 & \yes{} 0.57 & \yes{} 0.73 & \yes{} 0.59 & \yes{} 0.65 & \yes{} 0.80 & \yes{} 0.68 & \yes{} 0.73 \\
refuel-06 & \no{} 0.49 & \no{} 0.51 & \no{} 2.31 & \no{} 12.56 & \no{} 79.61 & \yes{} 39.15 & \yes{} 76.02 & \yes{} 58.03 \\
rocks-12 & \no{} 4.84 & \no{} 8.67 & \yes{} 58.39 & \yes{} 227.96 & \yes{} 508.86 & \yes{} 843.53 & \yes{} 1703.31 & \yes{} 1325.33 \\
\bottomrule
\end{tabular}
\caption{SMPMC Results for decision tree synthesis in \Cb.}
\label{tab:C2:SMPMC}
\end{table*}

\begin{table*}
\small\centering
\begin{tabular}{lrrrrrrrr}
\toprule
Benchmark & 1 & 3 & 5 & 7 & 9 & 11 & 13 & 15 \\
\midrule
drone-4-1 & NR & NR & NR & NR & NR & NR & NR & NR \\
drone-4-2 & NR & NR & NR & NR & NR & NR & NR & NR \\
drone-8-2 & NR & NR & NR & NR & NR & NR & NR & NR \\
grid-avoid-4-0.1 & \yes{} 0.70 & \yes{} 0.78 & \yes{} 0.78 & \yes{} 0.79 & \yes{} 0.79 & \yes{} 0.62 & \yes{} 0.82 & \yes{} 0.57 \\
maze-alex & \no{} 1.92 & \no{} 318.13 & NR & \yes{} 511.86 & \yes{} 1021.02 & \yes{} 1056.54 & \yes{} 1710.65 & \yes{} 1256.50 \\
nrp-8 & \yes{} 0.78 & \yes{} 0.70 & \yes{} 0.77 & \yes{} 0.86 & \yes{} 0.93 & \yes{} 0.94 & \yes{} 1.04 & \yes{} 1.04 \\
refuel-06 & \no{} 2.21 & \no{} 3.43 & \no{} 12.46 & \no{} 128.89 & \no{} 616.50 & \yes{} 509.62 & \yes{} 142.48 & \yes{} 367.59 \\
rocks-12 & NR & NR & NR & NR & NR & NR & NR & NR \\
\bottomrule
\end{tabular}
\caption{SMT(LRA) Results for decision tree synthesis in \Cb.}
\label{tab:C2:SMTLRA}
\end{table*}

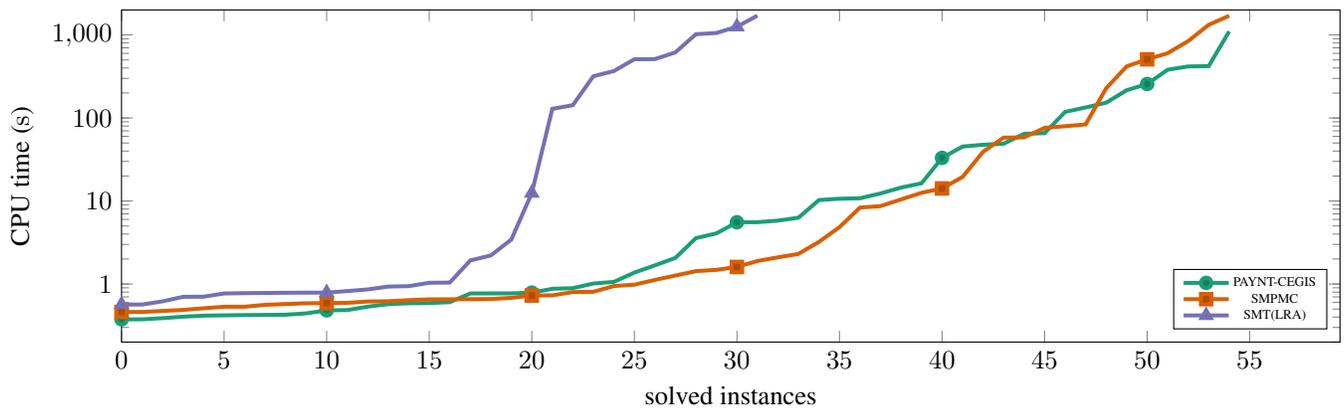
\begin{figure*}
    \pgfplotsset{
        width=\linewidth,
        height=6cm
    }
    \renewcommand{\pathtostuff}{experiments/full_plot_c2/}\input{experiments/full_plot_c2/quantile_processed.tex}
    \caption{Cactus Plot for \Cb.}
    \label{fig:C2:cactus}
\end{figure*}

\subsection{\Cc: Robust Synthesis}
\label{appendix:robust-synthesis}

\paragraph{Modifications of PAYNT-AR.}

We extend PAYNT-AR to support robust synthesis. In PAYNT-AR, the evaluation of the abstraction for a set of MCs (this usually represents a set of policies) includes one MDP model checking call. However, robust synthesis requires an evaluation against every environment, which leads to multiple model checking calls. We therefore get two loops: an outside loop that tries to find a satisfying policy, and an inner loop that evaluates it against the environments. We can leverage the PAYNT-AR in the outer loop, but we need to extend it with an inner loop for evaluation.
In the worst case, the inner loop needs to perform as many MDP model checking calls as there are environments. We try to improve on naive enumeration with two modifications: i) PAYNT-AR-1by1, which performs an enumeration in the inner loop with some tricks to improve naive enumeration, ii) PAYNT-AR-AR, which leverages another AR loop as the inner loop.

\paragraph{PAYNT-AR-1by1.}

PAYNT-AR-1by1 uses (non-naive) enumeration of the environments in the inner loop. It receives a set of policies from the outer loop and then finds the best policy out of this set on a fixed environment. It then checks whether some decisions of this policy are unreachable in this environment; if yes then these decisions can be removed from the policy. Each subsequent iteration in the inner loop then uses this policy and modifies it if needed i.e. the next iteration selects another environment, fixes the policy but some decisions may remain unfixed if they were removed previously, and finally checks the value of the policy in this environment, possibly fixing some of the open decisions. The inner loop continues until either i) it proves that the policy satisfies the bound in every environment and returns this policy, ii) it finds an environment where the policy does not work. In the latter case, the inner loop analyzes the difference between the policy it generated and the optimal policy on the unsatisfying environment and propagates this information to the outer loop which uses it to refine the abstraction. Additionally, when computing the optimal policy on some environment, we check if it satisfies the given bound; if not we can end the synthesis immediately, as we found an environment that has no satisfiable policy.

\paragraph{PAYNT-AR-AR.}

PAYNT-AR-AR is a straightforward improvement on the naive enumeration in the inner loop, however, as our experiments suggest, it has some drawbacks compared to the tricks in the improved enumeration we use in PAYNT-AR-1by1. The main drawback comes from the fact that in order to perform the inner AR loop, we need to fix the whole policy, which hinders the flexibility in the evaluation step as we compute the policy on one fixed environment and a fully fixed policy based on one environment might not work well on the set of environments. After fixing the policy, the inner AR loop over environments starts with the dual specification. The inner AR loop either i) proves that the policy works on all environments, or ii) finds an environment where the policy does not work. We then perform the same analysis as in PAYNT-AR-1by1, comparing two incompatible policies and propagating this information to the outer loop.

\paragraph{Experimental results.}

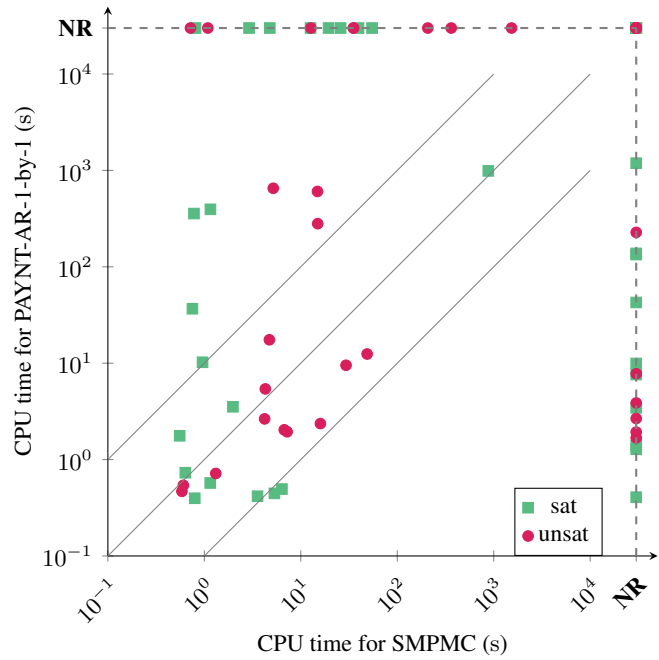
\begin{figure*}
\centering
    \pgfplotsset{
        width=0.5\linewidth,
        height=0.5\linewidth
    }
    \renewcommand{\pathtostuff}{experiments/robust-2025-07-22_18-59-50/}\input{experiments/robust-2025-07-22_18-59-50/scatter_PAYNT-AR-1-by-1.tex}
    \caption{\Cc: PAYNT-AR-1-by-1 vs SMPMC.}
    \label{fig:c3}
\end{figure*}

We show the results in \cref{tab:C3:sat} and \cref{tab:C3:unsat}.
We show a larger cactus plot in \cref{fig:quantilec3}.

\begin{table*}
\centering
\scriptsize
\begin{tabular}{lrrrrrrr}
\toprule
Benchmark & \texttt{AR-1-by-1} & \texttt{AR-AR} & \textbf{SMPMC} & \texttt{SMT(LRA)} & \#Assignments & \#States & SMPMC Iter. \\
\midrule
sat\_box-pushing & NR & NR & \textbf{54.85} & NR & $1.15\cdot 10^{18}$ & $20880$ & $3909$ \\
sat\_dpm & NR & $51.50$ & \textbf{19.41} & NR & $3.14\cdot 10^{10}$ & $19606$ & $242$ \\
sat\_maze & \textbf{0.45} & $0.56$ & $5.33$ & NR & $1.51\cdot 10^{8}$ & $195$ & $4423$ \\
sat\_mem\_1\_avoid-8-2 & NR & NR & NR & NR & NR & NR & NR \\
sat\_mem\_1\_avoid-8-2-bf & NR & NR & NR & NR & NR & NR & NR \\
sat\_mem\_1\_avoid-8-2-easy & NR & NR & NR & NR & NR & NR & NR \\
sat\_mem\_1\_avoid-8-2-easy-bf & NR & NR & NR & NR & $6.33\cdot 10^{2821}$ & $21169$ & NR \\
sat\_mem\_1\_catch-5 & \textbf{3.44} & NR & NR & NR & $4.41\cdot 10^{363}$ & $625$ & NR \\
sat\_mem\_1\_dpm & $394.85$ & $119.13$ & \textbf{1.16} & NR & $7.75\cdot 10^{9}$ & $737$ & $249$ \\
sat\_mem\_1\_dpm-switch-q10 & \textbf{42.65} & NR & NR & NR & $5.97\cdot 10^{275}$ & $1594$ & NR \\
sat\_mem\_1\_dpm-switch-q10 & $36.72$ & $3.54$ & \textbf{0.75} & NR & $2.55\cdot 10^{7}$ & $1594$ & $205$ \\
sat\_mem\_1\_dpm-switch-q10-big & $134.01$ & \textbf{3.72} & NR & NR & $7.69\cdot 10^{1503}$ & $8767$ & NR \\
sat\_mem\_1\_dpm-switch-q10-big-bf & $1188.33$ & \textbf{319.54} & NR & NR & $1.39\cdot 10^{1506}$ & $8767$ & NR \\
sat\_mem\_1\_dpm-switch-q10-bf & $988.35$ & NR & \textbf{878.77} & NR & $1.19\cdot 10^{277}$ & $1594$ & $40197$ \\
sat\_mem\_1\_obstacles-10-6-skip-easy & NR & NR & \textbf{12.63} & NR & $1.65\cdot 10^{63}$ & $718$ & $6853$ \\
sat\_mem\_1\_obstacles-10-6-skip-easy-bf & NR & NR & \textbf{25.83} & NR & $1.00\cdot 10^{66}$ & $783$ & $17531$ \\
sat\_mem\_1\_obstacles-demo & \textbf{0.41} & $0.42$ & NR & NR & $1.42\cdot 10^{22}$ & $83$ & NR \\
sat\_mem\_1\_pacman-6 & \textbf{1.28} & $9.32$ & NR & NR & $3.19\cdot 10^{760}$ & $1296$ & NR \\
sat\_mem\_1\_refuel-04 & $0.40$ & \textbf{0.36} & $0.80$ & NR & $1.32\cdot 10^{12}$ & $145$ & $551$ \\
sat\_mem\_1\_rocks-4-2 & \textbf{0.42} & $4.17$ & $3.56$ & NR & $6.42\cdot 10^{30}$ & $162$ & $4590$ \\
sat\_mem\_1\_rocks-6-4-small & \textbf{1.39} & $21.17$ & NR & NR & $3.32\cdot 10^{327}$ & $2716$ & NR \\
\textbf{sat\_mem\_1\_rocks-6-4} & NR & NR & NR & NR & $8.50\cdot 10^{328}$ & $2736$ & NR \\
sat\_mem\_1\_rocks-8-1 & \textbf{0.49} & $0.74$ & $6.37$ & NR & $2.18\cdot 10^{40}$ & $192$ & $7131$ \\
sat\_mem\_1\_rover & $356.65$ & $149.61$ & \textbf{0.78} & NR & $2.10\cdot 10^{6}$ & $902$ & $13$ \\
sat\_mem\_1\_rover-100 & $1.77$ & $1.45$ & \textbf{0.56} & NR & $6.91\cdot 10^{3}$ & $1669$ & $7$ \\
sat\_mem\_1\_rover-100-big & $0.73$ & \textbf{0.58} & $0.64$ & NR & $3.47\cdot 10^{62}$ & $1703$ & $5$ \\
sat\_mem\_1\_rover-100-big-bf & NR & NR & \textbf{0.81} & NR & $3.90\cdot 10^{67}$ & $1703$ & $18$ \\
sat\_mem\_1\_rover-1000 & $3.53$ & \textbf{1.93} & $1.98$ & NR & $1.86\cdot 10^{604}$ & $16986$ & $8$ \\
sat\_mem\_1\_rover-1000-bf & NR & NR & \textbf{4.77} & NR & $4.34\cdot 10^{608}$ & $17003$ & $19$ \\
sat\_mem\_1\_uav-operator-roz-workload & $9.93$ & \textbf{4.70} & NR & NR & $4.73\cdot 10^{1552}$ & $7903$ & NR \\
sat\_mem\_1\_uav-operator-roz-workload-bf & NR & NR & NR & NR & $4.85\cdot 10^{1758}$ & $9007$ & NR \\
sat\_mem\_1\_uav-roz & $7.63$ & \textbf{4.24} & NR & NR & $3.78\cdot 10^{1552}$ & $7903$ & NR \\
sat\_mem\_1\_uav-roz-bf & $137.77$ & \textbf{34.64} & NR & NR & $2.18\cdot 10^{3491}$ & $17927$ & NR \\
sat\_mem\_1\_virus-bf & NR & NR & \textbf{2.91} & NR & $5.01\cdot 10^{194}$ & $1678$ & $17$ \\
sat\_mem\_2\_dpm & NR & $290.38$ & \textbf{39.66} & NR & $5.15\cdot 10^{29}$ & $1474$ & $8097$ \\
sat\_mem\_2\_dpm-switch-q10 & NR & NR & NR & NR & $7.28\cdot 10^{36}$ & $3188$ & NR \\
sat\_mem\_2\_obstacles-8-3 & NR & NR & NR & NR & $8.61\cdot 10^{23}$ & $936$ & NR \\
sat\_mem\_2\_obstacles-8-6 & NR & NR & NR & NR & $8.81\cdot 10^{26}$ & $1138$ & NR \\
sat\_mem\_2\_refuel-04 & \textbf{0.57} & $0.65$ & $1.15$ & NR & $1.25\cdot 10^{69}$ & $289$ & $708$ \\
sat\_mem\_2\_rover-100 & $10.22$ & $5.66$ & \textbf{0.96} & NR & $3.04\cdot 10^{16}$ & $3338$ & $7$ \\
sat\_mem\_3\_obstacles-8-3 & NR & NR & NR & NR & $1.19\cdot 10^{47}$ & $1404$ & NR \\
sat\_mem\_3\_obstacles-8-6 & NR & NR & NR & NR & $2.53\cdot 10^{53}$ & $1707$ & NR \\
\bottomrule
\end{tabular}
\caption{\Cc{} experiments: SAT models (continued on next page).}
\label{tab:C3:sat}
\end{table*}

\begin{table*}
\centering
\scriptsize
\begin{tabular}{lrrrrrrr}
\toprule
Benchmark & \texttt{AR-1-by-1} & \texttt{AR-AR} & \textbf{SMPMC} & \texttt{SMT(LRA)} & \#Assignments & \#States & SMPMC Iter. \\
\midrule
unsat\_box-pushing & NR & NR & \textbf{362.54} & NR & $1.15\cdot 10^{18}$ & $20880$ & $8247$ \\
unsat\_dpm & NR & NR & \textbf{207.31} & NR & $3.14\cdot 10^{10}$ & $19606$ & $2363$ \\
unsat\_maze & $2.36$ & \textbf{0.57} & $15.97$ & NR & $1.51\cdot 10^{8}$ & $195$ & $8523$ \\
unsat\_mem\_1\_avoid & $9.52$ & \textbf{2.55} & $29.44$ & $1564.52$ & $4.00\cdot 10^{4}$ & $17567$ & $10$ \\
unsat\_mem\_1\_avoid-8-2 & NR & NR & NR & NR & NR & NR & NR \\
unsat\_mem\_1\_avoid-8-2-bf & NR & NR & NR & NR & NR & NR & NR \\
unsat\_mem\_1\_avoid-8-2-easy & \textbf{2.67} & NR & NR & NR & NR & NR & NR \\
unsat\_mem\_1\_avoid-8-2-easy-bf & NR & NR & NR & NR & $6.33\cdot 10^{2821}$ & $21169$ & NR \\
unsat\_mem\_1\_dpm & $604.66$ & NR & \textbf{14.90} & NR & $7.75\cdot 10^{9}$ & $737$ & $334$ \\
unsat\_mem\_1\_dpm-switch-q10 & NR & \textbf{0.59} & $35.07$ & NR & $5.97\cdot 10^{275}$ & $1594$ & $4747$ \\
unsat\_mem\_1\_dpm-switch-q10 & $652.65$ & $213.39$ & \textbf{5.18} & NR & $2.55\cdot 10^{7}$ & $1594$ & $125$ \\
unsat\_mem\_1\_dpm-switch-q10-big & NR & \textbf{0.99} & NR & NR & $7.69\cdot 10^{1503}$ & $8767$ & NR \\
unsat\_mem\_1\_dpm-switch-q10-big-bf & $227.51$ & \textbf{3.09} & NR & NR & $1.39\cdot 10^{1506}$ & $8767$ & NR \\
unsat\_mem\_1\_dpm-switch-q10-bf & NR & $0.77$ & \textbf{0.72} & NR & $1.19\cdot 10^{277}$ & $1594$ & $4$ \\
unsat\_mem\_1\_obstacles-10-6-skip-easy & NR & NR & NR & NR & $1.65\cdot 10^{63}$ & $718$ & NR \\
unsat\_mem\_1\_obstacles-10-6-skip-easy-bf & NR & NR & \textbf{35.26} & NR & $1.00\cdot 10^{66}$ & $783$ & $17970$ \\
unsat\_mem\_1\_obstacles-8-3 & $2.64$ & \textbf{0.41} & $4.22$ & $3.53$ & $2.99\cdot 10^{6}$ & $468$ & $25$ \\
unsat\_mem\_1\_obstacles-8-6 & $5.41$ & \textbf{0.40} & $4.29$ & $5.24$ & $1.19\cdot 10^{7}$ & $569$ & $25$ \\
unsat\_mem\_1\_obstacles-demo & \textbf{0.47} & $0.49$ & $0.59$ & $9.85$ & $1.42\cdot 10^{22}$ & $83$ & $51$ \\
unsat\_mem\_1\_obstacles-maze-4 & \textbf{2.04} & NR & $6.74$ & $58.61$ & $3.10\cdot 10^{4}$ & $31$ & $51$ \\
unsat\_mem\_1\_obstacles-maze-4-2 & \textbf{1.94} & NR & $7.22$ & $167.34$ & $6.20\cdot 10^{4}$ & $46$ & $47$ \\
unsat\_mem\_1\_rover-100-big & \textbf{0.54} & $0.58$ & $0.61$ & $300.65$ & $3.47\cdot 10^{62}$ & $1703$ & $3$ \\
unsat\_mem\_1\_rover-100-big-bf & NR & \textbf{0.87} & $1.09$ & NR & $3.90\cdot 10^{67}$ & $1703$ & $3$ \\
unsat\_mem\_1\_rover-1000 & \textbf{0.72} & $0.99$ & $1.31$ & NR & $1.86\cdot 10^{604}$ & $16986$ & $3$ \\
unsat\_mem\_1\_rover-1000-bf & $17.49$ & \textbf{4.70} & $4.72$ & NR & $4.34\cdot 10^{608}$ & $17003$ & $3$ \\
unsat\_mem\_1\_uav-operator-roz-workload & \textbf{1.68} & $1.71$ & NR & NR & $4.73\cdot 10^{1552}$ & $7903$ & NR \\
unsat\_mem\_1\_uav-operator-roz-workload-bf & $7.77$ & \textbf{2.90} & NR & NR & $4.85\cdot 10^{1758}$ & $9007$ & NR \\
unsat\_mem\_1\_uav-roz & \textbf{1.93} & $2.67$ & NR & NR & $3.78\cdot 10^{1552}$ & $7903$ & NR \\
unsat\_mem\_1\_uav-roz-bf & $3.87$ & \textbf{3.72} & NR & NR & $2.18\cdot 10^{3491}$ & $17927$ & NR \\
unsat\_mem\_1\_virus-bf & \textbf{12.45} & NR & $48.70$ & NR & $5.01\cdot 10^{194}$ & $1678$ & $3447$ \\
unsat\_mem\_2\_avoid & NR & NR & NR & NR & $4.19\cdot 10^{12}$ & $35133$ & NR \\
unsat\_mem\_2\_dpm & NR & NR & NR & NR & $5.15\cdot 10^{29}$ & $1474$ & NR \\
unsat\_mem\_2\_dpm-switch-q10 & NR & NR & \textbf{1530.89} & NR & $7.28\cdot 10^{36}$ & $3188$ & $284920$ \\
unsat\_mem\_2\_obstacles-8-3 & NR & NR & NR & NR & $8.61\cdot 10^{23}$ & $936$ & NR \\
unsat\_mem\_2\_obstacles-8-6 & NR & NR & NR & NR & $8.81\cdot 10^{26}$ & $1138$ & NR \\
unsat\_mem\_2\_obstacles-maze-4 & $280.11$ & NR & \textbf{14.97} & NR & $5.20\cdot 10^{11}$ & $62$ & $19181$ \\
unsat\_mem\_2\_obstacles-maze-4-2 & NR & NR & \textbf{12.70} & NR & $1.06\cdot 10^{15}$ & $92$ & $15087$ \\
unsat\_mem\_3\_obstacles-8-3 & NR & NR & NR & NR & $1.19\cdot 10^{47}$ & $1404$ & NR \\
unsat\_mem\_3\_obstacles-maze-4 & NR & NR & NR & NR & $9.68\cdot 10^{20}$ & $93$ & NR \\
unsat\_mem\_3\_obstacles-maze-4-2 & NR & NR & NR & NR & $2.78\cdot 10^{28}$ & $138$ & NR \\
\bottomrule
\end{tabular}
\caption{\Cc{} experiments: UNSAT models.}
\label{tab:C3:unsat}
\end{table*}

\begin{figure*}
    \pgfplotsset{
        width=\linewidth,
        height=6cm
    }
    \renewcommand{\pathtostuff}{experiments/robust-2025-07-22_18-59-50/}\input{experiments/robust-2025-07-22_18-59-50/quantile_processed.tex}
    \caption{Cactus Plot for \Cc.}
    \label{fig:quantilec3}
\end{figure*}
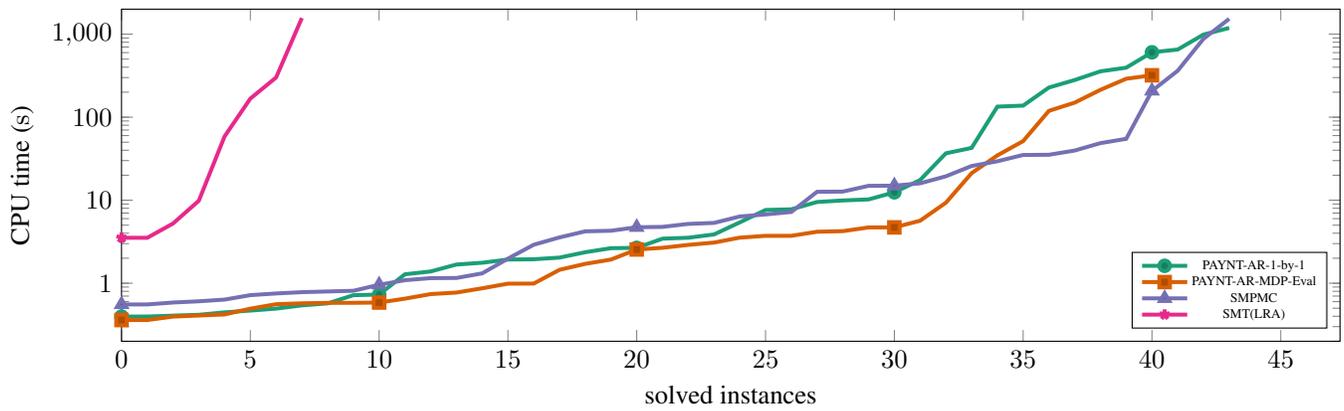

\subsection{\Cd: Plain Synthesis}

\paragraph{Experimental results.}
We show tables in \cref{tab:Q4:sat} and \cref{tab:Q4:unsat}.
We show a scatter plot of SMPMC against PAYNT-AR in \cref{fig:Q4:ar} and an enlarged cactus plot in \cref{fig:Q4:cactus}.

\begin{figure*}
\centering
    \pgfplotsset{
        width=0.5\linewidth,
        height=0.5\linewidth
    }
    \renewcommand{\pathtostuff}{experiments/simple-2025-07-29_21-15-45/}\input{experiments/simple-2025-07-29_21-15-45/scatter_Paynt-AR.tex}
    \caption{\Cd: PAYNT-AR vs SMPMC.}
    \label{fig:Q4:ar}
\end{figure*}
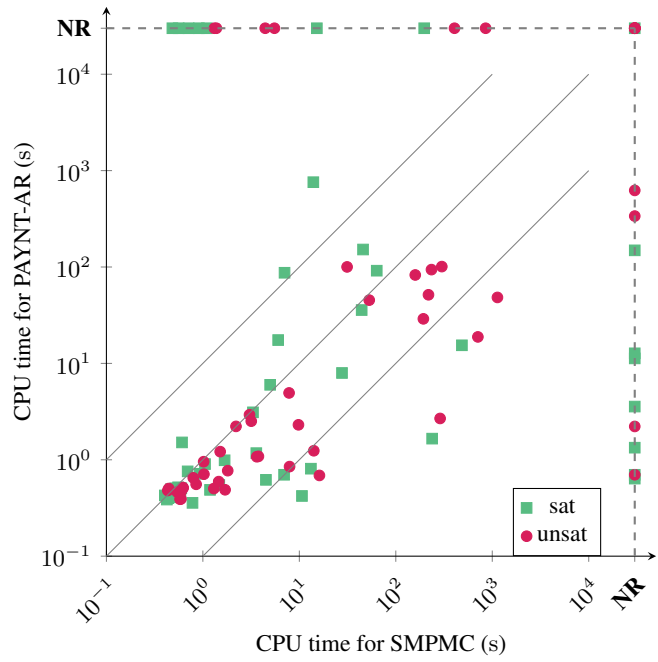
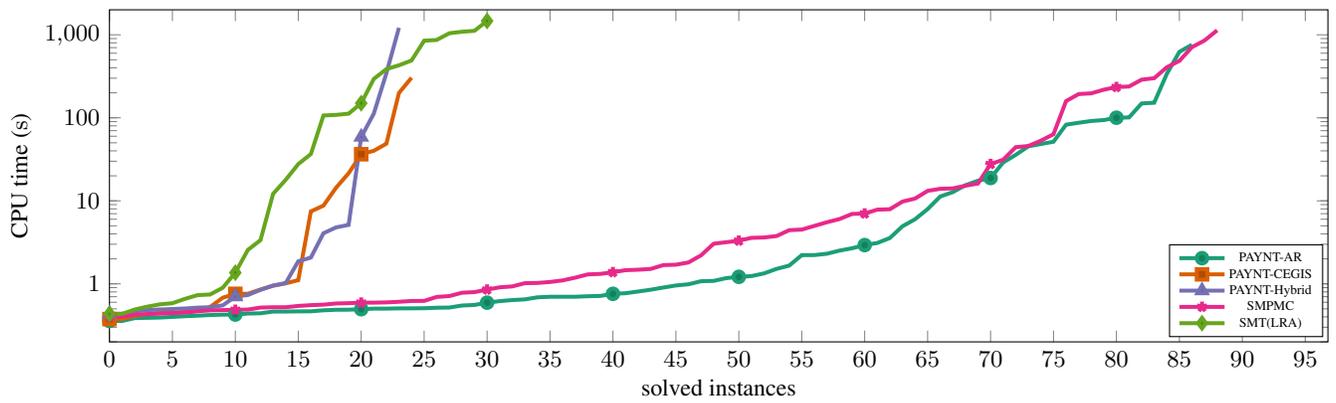
\begin{figure*}
    \pgfplotsset{
        width=\linewidth,
        height=6cm
    }
    \renewcommand{\pathtostuff}{experiments/simple-2025-07-29_21-15-45/}\input{experiments/simple-2025-07-29_21-15-45/quantile_processed.tex}
    \caption{Larger Cactus Plot for \Cd.}
    \label{fig:Q4:cactus}
\end{figure*}

\begin{table*}
\centering
\footnotesize
\begin{tabular}{lrrrrrrrr}
\toprule
Benchmark & \texttt{AR} & \texttt{CEGIS} & \texttt{Hybrid} & \textbf{SMPMC} & \texttt{SMT(LRA)} & \#Assignments & \#States & \#SMPMC Iter. \\
\midrule
sat\_box-pushing-1fsc & \textbf{3.11} & NR & NR & $3.31$ & NR & $1.05\cdot 10^{6}$ & $5220$ & $166$ \\
sat\_box-pushing-2fsc & \textbf{15.41} & NR & NR & $484.61$ & NR & $1.15\cdot 10^{18}$ & $20880$ & $28468$ \\
sat\_dpm & $151.89$ & $48.75$ & NR & \textbf{45.76} & NR & $4.30\cdot 10^{7}$ & $1652566$ & $2$ \\
sat\_grid & $1.51$ & $0.76$ & $4.07$ & \textbf{0.61} & $1087.32$ & $6.55\cdot 10^{4}$ & $289$ & $211$ \\
sat\_grid-10-sl-4fsc & NR & $0.76$ & NR & \textbf{0.71} & NR & $6.55\cdot 10^{4}$ & $2405$ & $20$ \\
sat\_grid-meet-sl-2fsc & $758.45$ & $40.07$ & NR & \textbf{13.97} & NR & $1.68\cdot 10^{7}$ & $9216$ & $1060$ \\
sat\_grid3x3-1fsc & \textbf{1.17} & NR & NR & $3.58$ & NR & $3.81\cdot 10^{12}$ & $6075$ & $257$ \\
sat\_grid3x3-2fsc & \textbf{7.95} & NR & NR & $27.70$ & NR & $1.00\cdot 10^{36}$ & $24300$ & $427$ \\
sat\_herman & \textbf{35.78} & $36.42$ & $1213.61$ & $44.29$ & NR & $3.12\cdot 10^{6}$ & $18753$ & $19$ \\
sat\_maze & NR & $14.36$ & NR & \textbf{0.56} & NR & $9.44\cdot 10^{6}$ & $183$ & $57$ \\
sat\_maze-sl-5 & NR & NR & NR & NR & NR & $7.95\cdot 10^{25}$ & $424$ & NR \\
sat\_mem\_1\_4x3-95 & \textbf{0.43} & NR & NR & $0.59$ & NR & $1.64\cdot 10^{4}$ & $22$ & $256$ \\
sat\_mem\_1\_4x5x2-95 & \textbf{0.39} & NR & NR & $0.42$ & $1.36$ & $1.02\cdot 10^{3}$ & $79$ & $2$ \\
sat\_mem\_1\_drone-4-1 & \textbf{3.56} & NR & NR & NR & NR & $6.16\cdot 10^{113}$ & $1226$ & NR \\
sat\_mem\_1\_drone-4-2 & \textbf{0.64} & NR & NR & NR & NR & $5.92\cdot 10^{225}$ & $1226$ & NR \\
sat\_mem\_1\_drone-8-2 & \textbf{12.73} & NR & NR & NR & NR & $3.02\cdot 10^{958}$ & $13042$ & NR \\
sat\_mem\_1\_grid-avoid-4-0.1 & \textbf{0.44} & $0.46$ & $0.46$ & $0.49$ & $0.53$ & $4.00\cdot 10^{0}$ & $17$ & $3$ \\
sat\_mem\_1\_hallway & $87.16$ & NR & NR & \textbf{7.03} & NR & $3.81\cdot 10^{12}$ & $1500$ & $179$ \\
sat\_mem\_1\_maze-alex & NR & $1.01$ & NR & \textbf{0.48} & $27.71$ & $1.64\cdot 10^{4}$ & $15$ & $46$ \\
sat\_mem\_1\_milos-aaai97 & \textbf{0.70} & NR & NR & $6.95$ & NR & $1.01\cdot 10^{7}$ & $165$ & $7413$ \\
sat\_mem\_1\_network & \textbf{0.40} & NR & NR & $0.44$ & NR & $6.40\cdot 10^{1}$ & $19$ & $19$ \\
sat\_mem\_1\_nrp-8 & $0.46$ & \textbf{0.41} & $0.52$ & $0.45$ & $0.58$ & $2.56\cdot 10^{2}$ & $125$ & $2$ \\
sat\_mem\_1\_query-s2 & $0.43$ & NR & NR & \textbf{0.40} & $293.81$ & $1.60\cdot 10^{1}$ & $36$ & $2$ \\
sat\_mem\_1\_query-s3 & \textbf{0.51} & NR & NR & $0.52$ & NR & $8.10\cdot 10^{1}$ & $108$ & $7$ \\
sat\_mem\_1\_refuel-06 & \textbf{0.62} & $304.01$ & $0.95$ & $4.51$ & $106.62$ & $9.03\cdot 10^{13}$ & $208$ & $5012$ \\
sat\_mem\_1\_refuel-20 & NR & \textbf{7.46} & NR & $15.20$ & NR & $2.99\cdot 10^{56}$ & $6834$ & $2396$ \\
sat\_mem\_1\_rocks-12 & \textbf{1.34} & NR & NR & NR & NR & $4.12\cdot 10^{1076}$ & $6553$ & NR \\
sat\_mem\_2\_4x3-95 & \textbf{0.42} & NR & NR & $10.65$ & NR & $2.15\cdot 10^{9}$ & $38$ & $14358$ \\
sat\_mem\_2\_4x5x2-95 & \textbf{0.36} & NR & NR & $0.78$ & $12.15$ & $4.19\cdot 10^{6}$ & $153$ & $524$ \\
sat\_mem\_2\_grid-avoid-4-0.1 & $0.52$ & \textbf{0.43} & $0.44$ & $0.55$ & $0.73$ & $1.28\cdot 10^{2}$ & $31$ & $14$ \\
sat\_mem\_2\_maze-alex & NR & $199.20$ & NR & \textbf{1.10} & NR & $1.07\cdot 10^{9}$ & $25$ & $831$ \\
sat\_mem\_2\_milos-aaai97 & NR & NR & NR & \textbf{0.91} & NR & $5.32\cdot 10^{19}$ & $328$ & $390$ \\
sat\_mem\_2\_network & \textbf{0.49} & NR & NR & $1.19$ & $490.06$ & $5.24\cdot 10^{5}$ & $36$ & $743$ \\
sat\_mem\_2\_nrp-8 & $0.47$ & \textbf{0.45} & NR & $0.48$ & $3.36$ & $2.88\cdot 10^{17}$ & $230$ & $2$ \\
sat\_mem\_2\_query-s2 & \textbf{0.41} & NR & NR & $0.52$ & NR & $1.31\cdot 10^{5}$ & $70$ & $24$ \\
sat\_mem\_2\_query-s3 & $6.00$ & NR & NR & \textbf{4.99} & NR & $3.36\cdot 10^{6}$ & $214$ & $267$ \\
sat\_mem\_2\_refuel-06 & \textbf{0.70} & NR & NR & NR & NR & $4.08\cdot 10^{42}$ & $394$ & NR \\
sat\_mem\_3\_4x3-95 & \textbf{1.66} & NR & NR & $238.16$ & NR & $3.21\cdot 10^{14}$ & $54$ & $144107$ \\
sat\_mem\_3\_4x5x2-95 & \textbf{0.81} & NR & NR & $13.18$ & $428.97$ & $1.55\cdot 10^{10}$ & $227$ & $15389$ \\
sat\_mem\_3\_grid-avoid-4-0.1 & \textbf{0.41} & $0.45$ & $0.50$ & $0.49$ & $2.56$ & $5.18\cdot 10^{3}$ & $45$ & $11$ \\
sat\_mem\_3\_maze-alex & NR & NR & NR & \textbf{197.13} & NR & $1.07\cdot 10^{14}$ & $35$ & $152265$ \\
sat\_mem\_3\_network & $17.46$ & NR & NR & \textbf{6.03} & NR & $1.55\cdot 10^{10}$ & $53$ & $7200$ \\
sat\_mem\_3\_query-s2 & \textbf{0.72} & NR & NR & $0.93$ & NR & $6.53\cdot 10^{9}$ & $104$ & $151$ \\
sat\_mem\_3\_query-s3 & $91.56$ & NR & NR & \textbf{63.47} & NR & $8.47\cdot 10^{11}$ & $320$ & $2016$ \\
sat\_mem\_3\_refuel-06 & \textbf{11.27} & NR & NR & NR & NR & $5.67\cdot 10^{76}$ & $580$ & NR \\
sat\_mem\_4\_4x3-95 & \textbf{148.91} & NR & NR & NR & NR & $7.38\cdot 10^{19}$ & $70$ & NR \\
sat\_mem\_4\_grid-avoid-4-0.1 & $0.90$ & \textbf{0.68} & $4.78$ & $1.05$ & $112.09$ & $2.62\cdot 10^{5}$ & $59$ & $549$ \\
sat\_pole & \textbf{0.99} & NR & NR & $1.68$ & NR & $1.13\cdot 10^{15}$ & $9241$ & $189$ \\
sat\_pole-res & $0.76$ & NR & NR & \textbf{0.70} & NR & $1.33\cdot 10^{6}$ & $6745$ & $25$ \\
sat\_recycling-2fsc & \textbf{0.44} & NR & NR & $0.52$ & NR & $1.68\cdot 10^{6}$ & $344$ & $32$ \\
sat\_recycling-3fsc & \textbf{0.51} & NR & NR & $0.60$ & NR & $2.82\cdot 10^{11}$ & $774$ & $51$ \\
sat\_refuel-06-res & \textbf{0.55} & NR & $0.70$ & $1.48$ & NR & $1.13\cdot 10^{41}$ & $877$ & $582$ \\
sat\_tiny\_rewards & $0.49$ & NR & NR & $0.47$ & \textbf{0.43} & $4.00\cdot 10^{0}$ & $3$ & $1$ \\
\bottomrule
\end{tabular}
\caption{\Cd{} experiments: SAT models (continued on next page).}
    \label{tab:Q4:sat}
\end{table*}

\begin{table*}
\centering
\footnotesize
\begin{tabular}{lrrrrrrrr}
\toprule
Benchmark & \texttt{AR} & \texttt{CEGIS} & \texttt{Hybrid} & \textbf{SMPMC} & \texttt{SMT(LRA)} & \#Assignments & \#States & \#SMPMC Iter. \\
\midrule
unsat\_box-pushing-1fsc & \textbf{2.93} & NR & NR & $3.04$ & NR & $1.05\cdot 10^{6}$ & $5220$ & $178$ \\
unsat\_box-pushing-2fsc & \textbf{48.35} & NR & NR & $1131.48$ & NR & $1.15\cdot 10^{18}$ & $20880$ & $62423$ \\
unsat\_dpm & \textbf{45.33} & NR & NR & $53.04$ & NR & $4.30\cdot 10^{7}$ & $1652566$ & $1$ \\
unsat\_grid & $2.22$ & \textbf{0.85} & $2.07$ & $2.20$ & $1116.51$ & $6.55\cdot 10^{4}$ & $289$ & $2572$ \\
unsat\_grid-10-sl-4fsc & NR & $21.49$ & NR & \textbf{4.43} & NR & $6.55\cdot 10^{4}$ & $2405$ & $1005$ \\
unsat\_grid-meet-sl-2fsc & $100.21$ & NR & $112.15$ & \textbf{31.26} & NR & $1.68\cdot 10^{7}$ & $9216$ & $1260$ \\
unsat\_grid3x3-1fsc & \textbf{0.96} & NR & NR & $1.02$ & NR & $3.81\cdot 10^{12}$ & $6075$ & $18$ \\
unsat\_grid3x3-2fsc & NR & NR & NR & \textbf{853.64} & NR & $1.00\cdot 10^{36}$ & $24300$ & $9392$ \\
unsat\_herman & $82.68$ & NR & \textbf{58.64} & $159.13$ & NR & $3.12\cdot 10^{6}$ & $18753$ & $3326$ \\
unsat\_maze & NR & NR & NR & \textbf{1.38} & NR & $9.44\cdot 10^{6}$ & $183$ & $780$ \\
unsat\_maze-sl-5 & NR & NR & NR & NR & NR & $7.95\cdot 10^{25}$ & $424$ & NR \\
unsat\_mem\_1\_4x3-95 & \textbf{0.49} & NR & NR & $1.70$ & NR & $1.64\cdot 10^{4}$ & $22$ & $1459$ \\
unsat\_mem\_1\_4x5x2-95 & \textbf{0.52} & NR & NR & $0.62$ & $36.58$ & $1.02\cdot 10^{3}$ & $79$ & $58$ \\
unsat\_mem\_1\_drone-4-1 & $2.22$ & NR & \textbf{1.87} & NR & NR & $6.16\cdot 10^{113}$ & $1226$ & NR \\
unsat\_mem\_1\_drone-4-2 & \textbf{0.70} & NR & $0.85$ & NR & NR & $5.92\cdot 10^{225}$ & $1226$ & NR \\
unsat\_mem\_1\_drone-8-2 & NR & NR & NR & NR & NR & $3.02\cdot 10^{958}$ & $13042$ & NR \\
unsat\_mem\_1\_grid-avoid-4-0.1 & $0.48$ & $0.53$ & $0.52$ & \textbf{0.43} & $0.57$ & $4.00\cdot 10^{0}$ & $17$ & $5$ \\
unsat\_mem\_1\_hallway & NR & NR & NR & NR & NR & $3.81\cdot 10^{12}$ & $1500$ & NR \\
unsat\_mem\_1\_maze-alex & NR & $8.75$ & NR & \textbf{1.32} & $848.29$ & $1.64\cdot 10^{4}$ & $15$ & $641$ \\
unsat\_mem\_1\_milos-aaai97 & \textbf{2.31} & NR & NR & $9.79$ & NR & $1.01\cdot 10^{7}$ & $165$ & $7929$ \\
unsat\_mem\_1\_network & \textbf{0.39} & NR & NR & $0.59$ & $0.74$ & $6.40\cdot 10^{1}$ & $19$ & $41$ \\
unsat\_mem\_1\_nrp-8 & $0.51$ & \textbf{0.38} & $0.55$ & $0.62$ & $0.66$ & $2.56\cdot 10^{2}$ & $125$ & $23$ \\
unsat\_mem\_1\_query-s2 & \textbf{0.47} & NR & NR & $0.59$ & $866.47$ & $1.60\cdot 10^{1}$ & $36$ & $22$ \\
unsat\_mem\_1\_query-s3 & \textbf{1.21} & NR & NR & $1.51$ & NR & $8.10\cdot 10^{1}$ & $108$ & $114$ \\
unsat\_mem\_1\_refuel-06 & $1.08$ & NR & \textbf{1.02} & $3.62$ & $385.94$ & $9.03\cdot 10^{13}$ & $208$ & $3824$ \\
unsat\_mem\_1\_refuel-20 & NR & NR & NR & NR & NR & $2.99\cdot 10^{56}$ & $6834$ & NR \\
unsat\_mem\_1\_rocks-12 & \textbf{337.59} & NR & $343.48$ & NR & NR & $4.12\cdot 10^{1076}$ & $6553$ & NR \\
unsat\_mem\_2\_4x3-95 & \textbf{0.69} & NR & NR & $16.16$ & NR & $2.15\cdot 10^{9}$ & $38$ & $18026$ \\
unsat\_mem\_2\_4x5x2-95 & \textbf{0.85} & NR & NR & $7.92$ & $149.86$ & $4.19\cdot 10^{6}$ & $153$ & $6270$ \\
unsat\_mem\_2\_grid-avoid-4-0.1 & $0.50$ & $0.47$ & $0.50$ & \textbf{0.44} & $0.90$ & $1.28\cdot 10^{2}$ & $31$ & $78$ \\
unsat\_mem\_2\_maze-alex & NR & NR & NR & \textbf{5.52} & NR & $1.07\cdot 10^{9}$ & $25$ & $6230$ \\
unsat\_mem\_2\_milos-aaai97 & \textbf{1.24} & NR & NR & $14.12$ & NR & $5.32\cdot 10^{19}$ & $328$ & $4581$ \\
unsat\_mem\_2\_network & \textbf{0.59} & NR & NR & $1.46$ & $1043.53$ & $5.24\cdot 10^{5}$ & $36$ & $878$ \\
unsat\_mem\_2\_nrp-8 & NR & NR & NR & $406.09$ & \textbf{108.38} & $2.88\cdot 10^{17}$ & $230$ & $365661$ \\
unsat\_mem\_2\_query-s2 & \textbf{0.77} & NR & NR & $1.81$ & NR & $1.31\cdot 10^{5}$ & $70$ & $325$ \\
unsat\_mem\_2\_query-s3 & \textbf{93.94} & NR & NR & $234.35$ & NR & $3.36\cdot 10^{6}$ & $214$ & $11373$ \\
unsat\_mem\_2\_refuel-06 & NR & NR & NR & NR & NR & $4.08\cdot 10^{42}$ & $394$ & NR \\
unsat\_mem\_3\_4x3-95 & \textbf{18.83} & NR & NR & $709.45$ & NR & $3.21\cdot 10^{14}$ & $54$ & $575661$ \\
unsat\_mem\_3\_4x5x2-95 & \textbf{2.68} & NR & NR & $288.07$ & NR & $1.55\cdot 10^{10}$ & $227$ & $195453$ \\
unsat\_mem\_3\_grid-avoid-4-0.1 & \textbf{0.50} & $0.95$ & $0.73$ & $1.30$ & $17.90$ & $5.18\cdot 10^{3}$ & $45$ & $903$ \\
unsat\_mem\_3\_maze-alex & \textbf{0.39} & NR & $0.49$ & $0.58$ & NR & $1.07\cdot 10^{14}$ & $35$ & $1$ \\
unsat\_mem\_3\_network & \textbf{100.98} & NR & NR & $299.90$ & NR & $1.55\cdot 10^{10}$ & $53$ & $149876$ \\
unsat\_mem\_3\_query-s2 & \textbf{28.99} & NR & NR & $193.16$ & NR & $6.53\cdot 10^{9}$ & $104$ & $41108$ \\
unsat\_mem\_3\_query-s3 & \textbf{51.50} & NR & NR & $218.15$ & NR & $8.47\cdot 10^{11}$ & $320$ & $2724$ \\
unsat\_mem\_3\_refuel-06 & NR & NR & NR & NR & NR & $5.67\cdot 10^{76}$ & $580$ & NR \\
unsat\_mem\_4\_4x3-95 & \textbf{623.27} & NR & NR & NR & NR & $7.38\cdot 10^{19}$ & $70$ & NR \\
unsat\_mem\_4\_grid-avoid-4-0.1 & $4.94$ & \textbf{1.11} & $5.13$ & $7.83$ & $1468.21$ & $2.62\cdot 10^{5}$ & $59$ & $7413$ \\
unsat\_pole & \textbf{1.08} & NR & NR & $3.77$ & NR & $1.13\cdot 10^{15}$ & $9241$ & $381$ \\
unsat\_pole-res & \textbf{0.71} & NR & NR & $1.03$ & NR & $1.33\cdot 10^{6}$ & $6745$ & $43$ \\
unsat\_recycling-2fsc & \textbf{0.56} & NR & NR & $0.85$ & NR & $1.68\cdot 10^{6}$ & $344$ & $253$ \\
unsat\_recycling-3fsc & \textbf{2.51} & NR & NR & $3.17$ & NR & $2.82\cdot 10^{11}$ & $774$ & $1485$ \\
unsat\_refuel-06-res & $0.65$ & NR & \textbf{0.48} & $0.80$ & NR & $1.13\cdot 10^{41}$ & $877$ & $33$ \\
unsat\_tiny\_rewards & \textbf{0.46} & NR & NR & $0.54$ & $0.49$ & $4.00\cdot 10^{0}$ & $3$ & $2$ \\\bottomrule
\end{tabular}
\caption{\Cd{} experiments: UNSAT models.}
    \label{tab:Q4:unsat}
\end{table*}

%% file: experiments/full_plot_c2/quantile_processed.tex
\begin{tikzpicture}
\begin{semilogyaxis}[
        /pgfplots/table/header=false,
        xlabel=solved instances,
        ylabel=CPU time (\second),
        xmin=0,
        ymin=0.2,
        ymax=2000,
        mark repeat=500,
        cycle multiindex* list={
            Dark2 \nextlist
            mark list
        },
        mark repeat=10,    %
        mark size=2pt,     %
        legend style={nodes={scale=0.5, transform shape}},
        legend entries={PAYNT-CEGIS,SMPMC,SMT(LRA),},
        every axis legend/.append style={at={(1,0)}, anchor=south east, outer xsep=5pt, outer ysep=5pt,},
        ]
        \foreach \tool in { comparison_trees.2025-07-29_15-15-10.results.PAYNT-CEGIS.Benchmarks.xml.bz2.quantile.csv,comparison_trees.2025-07-29_15-15-10.results.SMPMC.Benchmarks.xml.bz2.quantile.csv,comparison_trees.2025-07-29_15-15-10.results.SMTLRA.Benchmarks.xml.bz2.quantile.csv} {
            \addplot+ [line width=1.5pt] table[y index=5] {\pathtostuff\tool};
        }
\end{semilogyaxis}
\end{tikzpicture}

%% file: experiments/robust-2025-07-22_18-59-50/scatter_PAYNT-AR-1-by-1.tex
\newcommand{\scatterfile}{scatter_PAYNT-AR-1-by-1.table.csv}
\begin{tikzpicture}
\begin{loglogaxis}[
    xlabel={\small CPU time for SMPMC (\second)},
    ylabel={\small CPU time for PAYNT-AR-1-by-1 (\second)},
    label style={font=\small},
    xlabel near ticks,
    ylabel near ticks,
    xlabel shift={-1ex},
    ylabel shift={-1ex},
    xmin=0.1, xmax=50000,
    ymin=0.1, ymax=50000,
    domain=0.1:10001,
    clip mode=individual,
    axis equal image,
    axis lines=left, %
    xtick={0.1,1,10,100,1000,10000,30000},
    ytick={0.1,1,10,100,1000,10000,30000},
    xticklabels={$10^{-1}$,$10^{0}$,$10^{1}$,$10^{2}$,$10^{3}$,$10^{4}$,\textbf{NR}},
    yticklabels={$10^{-1}$,$10^{0}$,$10^{1}$,$10^{2}$,$10^{3}$,$10^{4}$,\textbf{NR}},
    xticklabel style={rotate=45, anchor=north east, font=\small},
    yticklabel style={font=\small},
    minor tick num=0, %
    legend style={at={(0.9, 0)},anchor=south east, font=\small}
    ]
    \addlegendimage{only marks, mark=square*, draw=pastelgreen, fill=pastelgreen}
    \addlegendimage{only marks, mark=*, draw=ruby, fill=ruby}
    \addplot+[only marks, scatter, scatter src=explicit symbolic, 
              scatter/classes={sat={mark=square*, draw=pastelgreen, fill=pastelgreen}, unsat={mark=*, draw=ruby, fill=ruby}}] table[
                 header=false,
                 skip first n=3,
                 x index=4,
                 y index=8,
                 meta index=1
            ] {\pathtostuff\scatterfile};
    \addlegendentry{sat}
    \addlegendentry{unsat}

    \addplot[gray, domain=0.1:10000] {x};
    \addplot[gray, domain=0.1:1000] {10*x};
    \addplot[gray, domain=1:10000] {x/10};

    \addplot[gray, thick, dashed, domain=0.1:30000] {30000}; %
    \addplot[gray, thick, dashed] coordinates {(30000,0.1) (30000,30000)}; %

\end{loglogaxis}
\end{tikzpicture}

%% file: experiments/robust-2025-07-22_18-59-50/quantile_processed.tex
\begin{tikzpicture}
\begin{semilogyaxis}[
        /pgfplots/table/header=false,
        xlabel=solved instances,
        ylabel=CPU time (\second),
        xmin=0,
        ymin=0.2,
        ymax=2000,
        mark repeat=500,
        cycle multiindex* list={
            Dark2 \nextlist
            mark list
        },
        mark repeat=10,    %
        mark size=2pt,     %
        legend style={nodes={scale=0.5, transform shape}},
        legend entries={PAYNT-AR-1-by-1,PAYNT-AR-MDP-Eval,SMPMC,SMT(LRA),},
        every axis legend/.append style={at={(1,0)}, anchor=south east, outer xsep=5pt, outer ysep=5pt,},
        ]
        \foreach \tool in {comparison_robust.2025-07-22_18-59-51.results.PAYNT-AR-1-by-1.xml.bz2.quantile.csv,comparison_robust.2025-07-22_18-59-51.results.PAYNT-AR-MDP-Eval.xml.bz2.quantile.csv,comparison_robust.2025-07-22_18-59-51.results.SMPMC.xml.bz2.quantile.csv,comparison_robust.2025-07-22_18-59-51.results.SMTLRA.xml.bz2.quantile.csv} {
            \addplot+ [line width=1.5pt] table[y index=5] {\pathtostuff\tool};
        }
\end{semilogyaxis}
\end{tikzpicture}

%% file: experiments/simple-2025-07-29_21-15-45/scatter_Paynt-AR.tex
\newcommand{\scatterfile}{scatter_PAYNT-AR.table.csv}
\begin{tikzpicture}
\begin{loglogaxis}[
    xlabel={\small CPU time for SMPMC (\second)},
    ylabel={\small CPU time for PAYNT-AR (\second)},
    label style={font=\small},
    xlabel near ticks,
    ylabel near ticks,
    xlabel shift={-1ex},
    ylabel shift={-1ex},
    xmin=0.1, xmax=50000,
    ymin=0.1, ymax=50000,
    domain=0.1:10001,
    clip mode=individual,
    axis equal image,
    axis lines=left, %
    xtick={0.1,1,10,100,1000,10000,30000},
    ytick={0.1,1,10,100,1000,10000,30000},
    xticklabels={$10^{-1}$,$10^{0}$,$10^{1}$,$10^{2}$,$10^{3}$,$10^{4}$,\textbf{NR}},
    yticklabels={$10^{-1}$,$10^{0}$,$10^{1}$,$10^{2}$,$10^{3}$,$10^{4}$,\textbf{NR}},
    xticklabel style={rotate=45, anchor=north east, font=\small},
    yticklabel style={font=\small},
    minor tick num=0, %
    legend style={at={(0.9, 0)},anchor=south east, font=\small}
    ]
    \addlegendimage{only marks, mark=square*, draw=pastelgreen, fill=pastelgreen}
    \addlegendimage{only marks, mark=*, draw=ruby, fill=ruby}
    \addplot+[only marks, scatter, scatter src=explicit symbolic, 
              scatter/classes={sat={mark=square*, draw=pastelgreen, fill=pastelgreen}, unsat={mark=*, draw=ruby, fill=ruby}}] table[
                 header=false,
                 skip first n=3,
                 x index=4,
                 y index=8,
                 meta index=1
            ] {\pathtostuff\scatterfile};
    \addlegendentry{sat}
    \addlegendentry{unsat}

    \addplot[gray, domain=0.1:10000] {x};
    \addplot[gray, domain=0.1:1000] {10*x};
    \addplot[gray, domain=1:10000] {x/10};

    \addplot[gray, thick, dashed, domain=0.1:30000] {30000}; %
    \addplot[gray, thick, dashed] coordinates {(30000,0.1) (30000,30000)}; %

\end{loglogaxis}
\end{tikzpicture}

%% file: paper.bbl
\begin{thebibliography}{43}
\providecommand{\natexlab}[1]{#1}

\bibitem[{Achterberg(2007)}]{DBLP:journals/disopt/Achterberg07}
Achterberg, T. 2007.
\newblock Conflict analysis in mixed integer programming.
\newblock \emph{Discret. Optim.}, 4(1): 4--20.

\bibitem[{Amato, Bonet, and Zilberstein(2010)}]{DBLP:conf/aaai/AmatoBZ10}
Amato, C.; Bonet, B.; and Zilberstein, S. 2010.
\newblock Finite-State Controllers Based on Mealy Machines for Centralized and
  Decentralized POMDPs.
\newblock In \emph{{AAAI}}, 1052--1058. {AAAI} Press.

\bibitem[{Andriushchenko et~al.(2023)Andriushchenko, Bork, {\v{C}}e{\v{s}}ka,
  Junges, Katoen, and Mac{\'a}k}]{andriushchenko23search}
Andriushchenko, R.; Bork, A.; {\v{C}}e{\v{s}}ka, M.; Junges, S.; Katoen,
  {\relax J{-}P}.; and Mac{\'a}k, F. 2023.
\newblock Search and Explore: Symbiotic Policy Synthesis in {POMDPs}.
\newblock In \emph{CAV}.
\newblock ISBN 978-3-031-37709-9.

\bibitem[{Andriushchenko et~al.(2024)Andriushchenko, {\v{C}}e{\v{s}}ka, Junges,
  and Mac{\'{a}}k}]{DBLP:conf/atva/AndriushchenkoCJM24}
Andriushchenko, R.; {\v{C}}e{\v{s}}ka, M.; Junges, S.; and Mac{\'{a}}k, F.
  2024.
\newblock Policies Grow on Trees: Model Checking Families of MDPs.
\newblock In \emph{{ATVA} {(2)}}, volume 15055 of \emph{{LNCS}}, 51--75.
  Springer.

\bibitem[{Andriushchenko et~al.(2025{\natexlab{a}})Andriushchenko,
  {\v{C}}e{\v{s}}ka, Junges, and Mac{\'a}k}]{andriushchenko2025small}
Andriushchenko, R.; {\v{C}}e{\v{s}}ka, M.; Junges, S.; and Mac{\'a}k, F.
  2025{\natexlab{a}}.
\newblock Small Decision Trees for MDPs with Deductive Synthesis.
\newblock In \emph{{CAV}}, 169--192. Springer.

\bibitem[{Andriushchenko et~al.(2025{\natexlab{b}})Andriushchenko,
  {\v{C}}e{\v{s}}ka, Mac{\'{a}}k, Junges, and
  Katoen}]{DBLP:journals/jair/AndriushchenkoCMJK25}
Andriushchenko, R.; {\v{C}}e{\v{s}}ka, M.; Mac{\'{a}}k, F.; Junges, S.; and
  Katoen, {\relax J{-}P}. 2025{\natexlab{b}}.
\newblock An Oracle-Guided Approach to Constrained Policy Synthesis Under
  Uncertainty.
\newblock \emph{J. Artif. Intell. Res.}, 82: 433--469.

\bibitem[{Ashok et~al.(2021)Ashok, Jackermeier, K{\v{r}}et{\'{\i}}nsk{\'{y}},
  Weinhuber, Weininger, and Yadav}]{DBLP:conf/tacas/AshokJKWWY21}
Ashok, P.; Jackermeier, M.; K{\v{r}}et{\'{\i}}nsk{\'{y}}, J.; Weinhuber, C.;
  Weininger, M.; and Yadav, M. 2021.
\newblock {dtControl} 2.0: Explainable Strategy Representation via Decision
  Tree Learning Steered by Experts.
\newblock In \emph{{TACAS}}.

\bibitem[{{Bayardo Jr.} and Schrag(1997)}]{DBLP:conf/aaai/BayardoS97}
{Bayardo Jr.}, R.~J.; and Schrag, R. 1997.
\newblock Using {CSP} Look-Back Techniques to Solve Real-World {SAT} Instances.
\newblock In \emph{{AAAI/IAAI}}, 203--208. {AAAI} Press / The {MIT} Press.

\bibitem[{Biere et~al.(2009)Biere, Heule, van Maaren, and
  Walsh}]{DBLP:series/faia/2009-185}
Biere, A.; Heule, M.; van Maaren, H.; and Walsh, T., eds. 2009.
\newblock \emph{Handbook of Satisfiability}, volume 185 of \emph{Frontiers in
  Artificial Intelligence and Applications}.
\newblock {IOS} Press.

\bibitem[{Bj{\o}rner, Eisenhofer, and
  Kov{\'{a}}cs(2023)}]{DBLP:conf/vmcai/BjornerEK23}
Bj{\o}rner, N.~S.; Eisenhofer, C.; and Kov{\'{a}}cs, L. 2023.
\newblock Satisfiability Modulo Custom Theories in {Z3}.
\newblock In \emph{{VMCAI}}, volume 13881 of \emph{{LNCS}}, 91--105. Springer.

\bibitem[{Bofill, Espasa, and Villaret(2017)}]{DBLP:conf/ijcai/BofillEV17}
Bofill, M.; Espasa, J.; and Villaret, M. 2017.
\newblock Relaxed Exists-Step Plans in Planning as {SMT}.
\newblock In \emph{{IJCAI}}, 563--570. ijcai.org.

\bibitem[{{\v{C}}e{\v{s}}ka et~al.(2021){\v{C}}e{\v{s}}ka, Hensel, Junges, and
  Katoen}]{vcevska2021counterexample}
{\v{C}}e{\v{s}}ka, M.; Hensel, C.; Junges, S.; and Katoen, {\relax J{-}P}.
  2021.
\newblock Counterexample-guided inductive synthesis for probabilistic systems.
\newblock \emph{Formal Aspects of Computing}, 33(4): 637--667.

\bibitem[{Chades et~al.(2012)Chades, Carwardine, Martin, Nicol, Sabbadin, and
  Buffet}]{DBLP:conf/aaai/ChadesCMNSB12}
Chades, I.; Carwardine, J.; Martin, T.~G.; Nicol, S.; Sabbadin, R.; and Buffet,
  O. 2012.
\newblock {MOMDPs}: {A} Solution for Modelling Adaptive Management Problems.
\newblock In \emph{{AAAI}}. {AAAI} Press.

\bibitem[{Chatterjee et~al.(2020)Chatterjee, Chmel{\'{\i}}k, Karkhanis,
  Novotn{\'{y}}, and Royer}]{DBLP:conf/aips/ChatterjeeCK0R20}
Chatterjee, K.; Chmel{\'{\i}}k, M.; Karkhanis, D.; Novotn{\'{y}}, P.; and
  Royer, A. 2020.
\newblock Multiple-Environment {Markov} Decision Processes: Efficient Analysis
  and Applications.
\newblock In \emph{{ICAPS}}, 48--56. {AAAI} Press.

\bibitem[{Chatterjee, Doyen, and
  Henzinger(2010)}]{DBLP:conf/mfcs/ChatterjeeDH10}
Chatterjee, K.; Doyen, L.; and Henzinger, T.~A. 2010.
\newblock Qualitative Analysis of Partially-Observable Markov Decision
  Processes.
\newblock In \emph{MFCS}, volume 6281 of \emph{Lecture Notes in Computer
  Science}, 258--269. Springer.

\bibitem[{Chrszon et~al.(2018)Chrszon, Dubslaff, Kl{\"{u}}ppelholz, and
  Baier}]{DBLP:journals/fac/ChrszonDKB18}
Chrszon, P.; Dubslaff, C.; Kl{\"{u}}ppelholz, S.; and Baier, C. 2018.
\newblock ProFeat: feature-oriented engineering for family-based probabilistic
  model checking.
\newblock \emph{Formal Aspects Comput.}, 30(1): 45--75.

\bibitem[{Cubuktepe et~al.(2021)Cubuktepe, Jansen, Junges, Marandi, Suilen, and
  Topcu}]{DBLP:conf/aaai/Cubuktepe0JMST21}
Cubuktepe, M.; Jansen, N.; Junges, S.; Marandi, A.; Suilen, M.; and Topcu, U.
  2021.
\newblock Robust Finite-State Controllers for Uncertain {POMDP}s.
\newblock In \emph{{AAAI}}, 11792--11800. {AAAI} Press.

\bibitem[{de~Moura and Bj{\o}rner(2007)}]{DBLP:conf/cade/MouraB07}
de~Moura, L.~M.; and Bj{\o}rner, N.~S. 2007.
\newblock Efficient E-Matching for {SMT} Solvers.
\newblock In \emph{{CADE}}, volume 4603 of \emph{Lecture Notes in Computer
  Science}, 183--198. Springer.

\bibitem[{de~Moura and Bj{\o}rner(2008)}]{DBLP:conf/tacas/MouraB08}
de~Moura, L.~M.; and Bj{\o}rner, N.~S. 2008.
\newblock {Z3:} An Efficient {SMT} Solver.
\newblock In \emph{{TACAS}}, volume 4963 of \emph{{LNCS}}, 337--340. Springer.

\bibitem[{Eisenhofer et~al.(2023)Eisenhofer, Alassaf, Rawson, and
  Kov{\'{a}}cs}]{DBLP:conf/tableaux/EisenhoferARK23}
Eisenhofer, C.; Alassaf, R.; Rawson, M.; and Kov{\'{a}}cs, L. 2023.
\newblock Non-Classical Logics in Satisfiability Modulo Theories.
\newblock In \emph{{TABLEAUX}}, volume 14278 of \emph{Lecture Notes in Computer
  Science}, 24--36. Springer.

\bibitem[{Galesloot et~al.(2025)Galesloot, Andriushchenko, Ceska, Junges, and
  Jansen}]{hmpomdps}
Galesloot, M. F.~L.; Andriushchenko, R.; Ceska, M.; Junges, S.; and Jansen, N.
  2025.
\newblock Robust Finite-Memory Policy Gradients for Hidden-Model POMDPs.
\newblock In \emph{{IJCAI}}, 8518--8526. ijcai.org.

\bibitem[{Ge and de~Moura(2009)}]{DBLP:conf/cav/GeM09}
Ge, Y.; and de~Moura, L.~M. 2009.
\newblock Complete Instantiation for Quantified Formulas in Satisfiabiliby
  Modulo Theories.
\newblock In \emph{{CAV}}, volume 5643 of \emph{{LNCS}}, 306--320. Springer.

\bibitem[{Ghezzi and Sharifloo(2013)}]{DBLP:journals/infsof/GhezziS13}
Ghezzi, C.; and Sharifloo, A.~M. 2013.
\newblock Model-based verification of quantitative non-functional properties
  for software product lines.
\newblock \emph{Inf. Softw. Technol.}, 55(3): 508--524.

\bibitem[{Hartmanns et~al.(2023)Hartmanns, Junges, Quatmann, and
  Weininger}]{DBLP:conf/tacas/HartmannsJQW23}
Hartmanns, A.; Junges, S.; Quatmann, T.; and Weininger, M. 2023.
\newblock A Practitioner's Guide to {MDP} Model Checking Algorithms.
\newblock In \emph{{TACAS} {(1)}}, volume 13993 of \emph{{LNCS}}, 469--488.
  Springer.

\bibitem[{Hartmanns and Kaminski(2020)}]{DBLP:conf/cav/HartmannsK20}
Hartmanns, A.; and Kaminski, B.~L. 2020.
\newblock Optimistic Value Iteration.
\newblock In \emph{{CAV} {(2)}}, volume 12225 of \emph{{LNCS}}, 488--511.
  Springer.

\bibitem[{Hensel et~al.(2022)Hensel, Junges, Katoen, Quatmann, and
  Volk}]{DBLP:journals/sttt/HenselJKQV22}
Hensel, C.; Junges, S.; Katoen, {\relax J{-}P}.; Quatmann, T.; and Volk, M.
  2022.
\newblock The probabilistic model checker Storm.
\newblock \emph{Int. J. Softw. Tools Technol. Transf.}, 24(4): 589--610.

\bibitem[{Ho et~al.(2024)Ho, Feather, Rossi, Sunberg, and
  Lahijanian}]{DBLP:conf/uai/HoFRSL24}
Ho, Q.~H.; Feather, M.~S.; Rossi, F.; Sunberg, Z.; and Lahijanian, M. 2024.
\newblock Sound Heuristic Search Value Iteration for Undiscounted POMDPs with
  Reachability Objectives.
\newblock In \emph{UAI}, volume 244 of \emph{Proceedings of Machine Learning
  Research}, 1681--1697. {PMLR}.

\bibitem[{Hooker(2004)}]{DBLP:conf/cp/Hooker04}
Hooker, J.~N. 2004.
\newblock A Hybrid Method for Planning and Scheduling.
\newblock In Wallace, M., ed., \emph{Principles and Practice of Constraint
  Programming - {CP} 2004, 10th International Conference, {CP} 2004, Toronto,
  Canada, September 27 - October 1, 2004, Proceedings}, volume 3258 of
  \emph{Lecture Notes in Computer Science}, 305--316. Springer.

\bibitem[{Kondylidou, Reynolds, and Blanchette(2025)}]{modelbasedfastenum}
Kondylidou, L.; Reynolds, A.; and Blanchette, J. 2025.
\newblock Augmenting Model-Based Instantiation with Fast Enumeration.
\newblock In \emph{TACAS}, 85--103. Springer.
\newblock ISBN 978-3-031-90643-5.

\bibitem[{Kumar and Zilberstein(2015{\natexlab{a}})}]{DBLP:conf/aips/KumarZ15}
Kumar, A.; and Zilberstein, S. 2015{\natexlab{a}}.
\newblock History-Based Controller Design and Optimization for Partially
  Observable MDPs.
\newblock In \emph{ICAPS}, 156--164. {AAAI} Press.

\bibitem[{Kumar and Zilberstein(2015{\natexlab{b}})}]{kumar2015history}
Kumar, A.; and Zilberstein, S. 2015{\natexlab{b}}.
\newblock History-based controller design and optimization for partially
  observable {MDP}s.
\newblock In \emph{ICAPS}, volume~25, 156--164.

\bibitem[{Kurniawati, Hsu, and Lee(2008)}]{DBLP:conf/rss/KurniawatiHL08}
Kurniawati, H.; Hsu, D.; and Lee, W.~S. 2008.
\newblock {SARSOP:} Efficient Point-Based {POMDP} Planning by Approximating
  Optimally Reachable Belief Spaces.
\newblock In \emph{Robotics: Science and Systems IV}. The {MIT} Press.

\bibitem[{Lin et~al.(2024)Lin, Xue, Deng, and Ye}]{DBLP:conf/icml/LinXD024}
Lin, Z.; Xue, C.; Deng, Q.; and Ye, Y. 2024.
\newblock A Single-Loop Robust Policy Gradient Method for Robust {Markov}
  Decision Processes.
\newblock In \emph{{ICML}}. OpenReview.net.

\bibitem[{Nakao, Jiang, and Shen(2021)}]{DBLP:journals/siamjo/NakaoJS21}
Nakao, H.; Jiang, R.; and Shen, S. 2021.
\newblock Distributionally Robust Partially Observable {Markov} Decision
  Process with Moment-Based Ambiguity.
\newblock \emph{{SIAM} J. Optim.}, 31(1): 461--488.

\bibitem[{Narodytska et~al.(2018)Narodytska, Ignatiev, Pereira, and
  Marques{-}Silva}]{DBLP:conf/ijcai/NarodytskaIPM18}
Narodytska, N.; Ignatiev, A.; Pereira, F.; and Marques{-}Silva, J. 2018.
\newblock Learning Optimal Decision Trees with {SAT}.
\newblock In \emph{IJCAI}, 1362--1368. ijcai.org.

\bibitem[{Nieuwenhuis, Oliveras, and
  Tinelli(2006)}]{DBLP:journals/jacm/NieuwenhuisOT06}
Nieuwenhuis, R.; Oliveras, A.; and Tinelli, C. 2006.
\newblock Solving {SAT} and {SAT} Modulo Theories: From an abstract
  Davis--Putnam--Logemann--Loveland procedure to DPLL(\emph{T}).
\newblock \emph{J. {ACM}}, 53(6): 937--977.

\bibitem[{Puterman(1994)}]{Put94}
Puterman, M.~L. 1994.
\newblock \emph{Markov Decision Processes: Discrete Stochastic Dynamic
  Programming}.
\newblock Wiley Series in Probability and Statistics. Wiley.

\bibitem[{Shoukry et~al.(2017)Shoukry, Nuzzo, Sangiovanni{-}Vincentelli,
  Seshia, Pappas, and Tabuada}]{DBLP:conf/hybrid/ShoukryNSSPT17}
Shoukry, Y.; Nuzzo, P.; Sangiovanni{-}Vincentelli, A.~L.; Seshia, S.~A.;
  Pappas, G.~J.; and Tabuada, P. 2017.
\newblock {SMC:} Satisfiability Modulo Convex Optimization.
\newblock In \emph{{HSCC}}, 19--28. {ACM}.

\bibitem[{Silva and Sakallah(1999)}]{DBLP:journals/tc/Marques-SilvaS99}
Silva, J. P.~M.; and Sakallah, K.~A. 1999.
\newblock {GRASP:} {A} Search Algorithm for Propositional Satisfiability.
\newblock \emph{{IEEE} Trans. Computers}, 48(5): 506--521.

\bibitem[{Smallwood and Sondik(1973)}]{smallwood1973optimal}
Smallwood, R.~D.; and Sondik, E.~J. 1973.
\newblock The optimal control of partially observable {M}arkov processes over a
  finite horizon.
\newblock \emph{Oper. Res.}, 21(5): 1071--1088.

\bibitem[{Staus et~al.(2025)Staus, Komusiewicz, Sommer, and
  Sorge}]{DBLP:conf/aaai/StausKSS25}
Staus, L.~P.; Komusiewicz, C.; Sommer, F.; and Sorge, M. 2025.
\newblock Witty: An Efficient Solver for Computing Minimum-Size Decision Trees.
\newblock In \emph{AAAI}, 20584--20591. {AAAI} Press.

\bibitem[{van~der Vegt, Jansen, and Junges(2023)}]{DBLP:conf/tacas/VegtJJ23}
van~der Vegt, M.; Jansen, N.; and Junges, S. 2023.
\newblock Robust Almost-Sure Reachability in Multi-Environment {MDP}s.
\newblock In \emph{{TACAS} {(1)}}, volume 13993 of \emph{{LNCS}}, 508--526.
  Springer.

\bibitem[{Vos and Verwer(2023)}]{DBLP:conf/ijcai/VosV23}
Vos, D.; and Verwer, S. 2023.
\newblock Optimal Decision Tree Policies for Markov Decision Processes.
\newblock In \emph{{IJCAI}}, 5457--5465. ijcai.org.

\end{thebibliography}
